\documentclass[11pt]{article}%
\usepackage{fullpage}
\usepackage{amsmath}
\usepackage{amsfonts}
\usepackage{amssymb}
\usepackage{graphicx}
\usepackage[colorlinks=true,linkcolor=blue,citecolor=red,plainpages=false,pdfpagelabels]%
{hyperref}%
\providecommand{\U}[1]{\protect\rule{.1in}{.1in}}
\newtheorem{theorem}{Theorem}

\newtheorem{conjecture}[theorem]{Conjecture}

\newtheorem{lemma}[theorem]{Lemma}

\newtheorem{proposition}[theorem]{Proposition}
\newtheorem{question}[theorem]{Question}
\newtheorem{remark}[theorem]{Remark}

\newenvironment{proof}[1][Proof]{\noindent\textbf{#1.} }{\ \rule{0.5em}{0.5em}}

\def\U{\mathrm{U}}

\numberwithin{equation}{section}
\begin{document}

\title{\textbf{Applications of position-based coding to classical communication over
quantum channels}}
\author{Haoyu Qi\thanks{Hearne Institute for Theoretical Physics and Department of
Physics \& Astronomy, Louisiana State University, Baton Rouge, Louisiana
70803, USA}
\and Qingle Wang\footnotemark[1] \thanks{State Key Laboratory of Networking and
Switching Technology, Beijing University of Posts and Telecommunications,
Beijing 100876, China}
\and Mark M. Wilde\footnotemark[1] \thanks{Center for Computation and Technology,
Louisiana State University, Baton Rouge, Louisiana 70803, USA}}
\maketitle

\begin{abstract}
Recently, a coding technique called \textit{position-based coding} has been
used to establish achievability statements for various kinds of classical
communication protocols that use quantum channels. In the present paper, we
apply this technique in the entanglement-assisted setting in order to
establish lower bounds for error exponents, lower bounds on the second-order
coding rate, and one-shot lower bounds. We also demonstrate that
position-based coding can be a powerful tool for analyzing other communication
settings. In particular, we reduce the quantum simultaneous decoding
conjecture for entanglement-assisted or unassisted communication over a
quantum multiple access channel to open questions in multiple quantum
hypothesis testing. We then determine achievable rate regions for
entanglement-assisted or unassisted classical communication over a quantum
multiple-access channel, when using a particular quantum simultaneous decoder.
The achievable rate regions given in this latter case are generally
suboptimal, involving differences of R\'{e}nyi-2 entropies and conditional
quantum entropies.

\end{abstract}

\section{Introduction}

Understanding optimal rates for classical communication over both
point-to-point quantum channels and quantum network channels are fundamental
tasks in quantum Shannon theory (see, e.g.,
\cite{H12,Wat16,wilde2011classical,W17}). Early developments of quantum
Shannon theory are based on the assumption that the information is transmitted
over an arbitrarily large number of independent and identically distributed
(i.i.d.)~uses of a given quantum channel. By taking advantage of this
assumption, general formulas have been established for capacities of various
communication protocols, with or without preshared entanglement. When a sender
and receiver do not share entanglement before communication begins, it is
known that the Holevo information of a quantum channel is an achievable rate
for classical communication \cite{holevo1998capacity,schumacher1997sending}.
Regularizing the Holevo information leads to a multi-letter formula that
characterizes the capacity for this task. Regarding communication over quantum
network channels, an achievable rate region for classical communication over
quantum multiple-access channels was given in \cite{winter2001capacity} and
regularizing it leads to a characterization of the capacity region for this
task. However, only inner bounds on the capacity region for general broadcast
channels are known
\cite{yard2011quantum,radhakrishnan2016one,savov2015classical}, except when
the quantum broadcast channel is a particular kind of degraded channel
\cite{wang2016hadamard}. When there is entanglement shared between the
communicating parties, many scenarios have also been studied, including
classical communication over point-to-point quantum channels
\cite{bennett2002entanglement,Hol01a}, quantum multiple-access channels
\cite{hsieh2008entanglement,Xu2013}, and quantum broadcast channels
\cite{dupuis2010father,wang2016hadamard}.

Although channel capacity gives a fundamental characterization of the
communication capabilities of a quantum channel, many practically important
properties of quantum channels are not captured by this quantity. To close
this gap, several works have focused on the study of refined notions of
capacity, including error exponents
\cite{BH98,thesis1999winter,H00,hayashi2007error,dalai2013lower,DW14} and
second-order asymptotics \cite{tomamichel2013second,wilde2016second,DTW14}.
The latter works build upon strong connections between hypothesis testing and
coding, as considered in \cite{H09,polyanskiy10,wang2012one}. The refined
characterizations of capacity are of importance for regimes of practical
interest, in which a limited number of uses of a quantum channel are
available. Complementary to these developments, to go beyond the
i.i.d.~assumption, many works have been dedicated to the one-shot formalism
\cite{wang2012one,datta2013one,datta2013one-EA,matthews2014finite} and the
information-spectrum approach
\cite{hayashi2003generalCapacity,Hay06,bowen2006beyond}, with very few
assumptions made on the structure of quantum channels.

In a recent work \cite{anshu2017one}, a technique called position-based coding
was developed in order to give one-shot achievability bounds for various
classical communication protocols that use entanglement assistance. This was
then extended to the cases of unassisted classical communication and private
classical communication \cite{wilde2017position}. The method of position-based
coding is a derivative of the well known and long studied coding technique
called pulse position modulation (see, e.g., \cite{verdu1990channel,eltit03}).
In pulse position modulation, a sender encodes a message by placing a pulse in
one slot and having no pulse in the other available slots. Each slot is then
communicated over the channel, one by one. The receiver can decode well if he
can distinguish \textquotedblleft pulse\textquotedblright\ from
\textquotedblleft no pulse.\textquotedblright\ Position-based coding borrows
this idea: a sender and receiver are allowed to share many copies of a
bipartite quantum state. The sender encodes a message by sending a share of
one of the bipartite states through a channel to the receiver. From the
receiver's perspective, only one share of his systems will be correlated with
the channel output (the pulse in one slot), while the others will have no
correlation (no pulse in the other slots). So if the receiver can distinguish
\textquotedblleft pulse\textquotedblright\ from \textquotedblleft no
pulse\textquotedblright\ in this context, then he will be able to decode well,
just as is the case in pulse position modulation. The authors of
\cite{anshu2017one}\ applied the position-based coding technique to a number
of problems that have already been addressed in the literature, including
point-to-point entanglement-assisted communication
\cite{bennett2002entanglement,Hol01a}, entanglement-assisted coding with side
information at the transmitter \cite{D09}, and entanglement-assisted
communication over broadcast channels \cite{dupuis2010father}.

In the present paper, we use position-based coding to establish several new
results. First, we establish lower bounds on the entanglement-assisted error
exponent and on the one-shot entanglement-assisted capacity. The latter
improves slightly upon the result from \cite{anshu2017one} and in turn gives a
simpler proof of one of the main results of \cite{DTW14}, i.e., a lower bound
on the second-order coding rate for entanglement-assisted classical
communication. We then turn to communication over quantum multiple-access
channels when using a quantum simultaneous decoder, considering both cases of
entanglement assistance and no assistance. The quantum simultaneous decoding
conjecture from \cite{fawzi2012classical,Wil11}\ stands as one of the most
important open problems in network quantum information theory. Here we report
progress on this conjecture and connect it to some open questions from
\cite{AM14,BHOS14}\ in multiple quantum hypothesis testing. At the same time,
we give new achievable rate regions for entanglement-assisted classical
communication over multiple-access channels, where the bounds on achievable
rates are expressed as a difference of a R\'{e}nyi entropy of order two and a
conditional quantum entropy.

This paper is organized as follows. We first summarize relevant definitions
and lemmas in Section~\ref{sec:preliminary}. In this section, we also prove
Proposition~\ref{prop:ineq-hypo-renyi}, which relates the hypothesis testing
relative entropy to the quantum R\'{e}nyi relative entropy and is an
interesting counterpart to \cite[Lemma~5]{CMW14}. In
Section~\ref{sec:EA-point-to-point}, we consider entanglement-assisted
point-to-point classical communication. By using position-based coding, we
establish a lower bound on the entanglement-assisted error exponent. We also
establish a lower bound on the one-shot entanglement-assisted capacity in
terms of hypothesis-testing mutual information and state how it is close to a
previously known upper bound from \cite{matthews2014finite}. At the same time,
we provide a simpler proof of the upper bound from \cite{matthews2014finite},
which follows from a lemma regarding generalized mutual information. Based
upon this one-shot lower bound, we then rederive a lower bound on the
second-order coding rate for entanglement-assisted communication with a proof
that is arguably simpler than that given in \cite{DTW14}. In
Section~\ref{sec:EA-MAC}, we apply position-based coding to
entanglement-assisted classical communication over multiple-access channels
and establish an explicit link to multiple quantum hypothesis testing. We give
an achievable rate region for i.i.d.~channels by using techniques from the
theory of quantum typicality. We demonstrate the power of position-based
coding technique in unassisted classical communication in
Section~\ref{sec:UA-MAC}, by considering classical communication over
multiple-access channel. We explicitly show how to derandomize a
randomness-assisted protocol. In Section~\ref{sec:disucssion}, we tie open
questions in multiple quantum hypothesis testing to quantum simultaneous
decoding for the quantum multiple-access channel. Finally, we summarize our
main results and discuss open questions in Section~\ref{sec:conclusion}.

\section{Preliminaries}

\label{sec:preliminary}

\textbf{Trace distance, fidelity, and gentle measurement.} Let $\mathcal{D}%
(\mathcal{H})$ denote the set of density operators (positive semi-definite
operators with unit trace) acting on a Hilbert space~$\mathcal{H}$. The
\textit{trace distance} between two density operators $\rho,\sigma
\in\mathcal{D}(\mathcal{H})$ is equal to $\Vert\rho-\sigma\Vert_{1}$, where
$\Vert A\Vert_{1}\equiv\operatorname{Tr}\{\sqrt{A^{\dagger}A}\}$. Another
quantity to measure the closeness between two quantum states is the
\textit{fidelity}, defined as $F(\rho,\sigma)\equiv\Vert\sqrt{\rho}%
\sqrt{\sigma}\Vert_{1}^{2}$ \cite{uhlmann1976transition}. Two inequalities
relating trace distance and quantum measurement operators are as follows:

\begin{lemma}
(Gentle operator \cite{itit1999winter,ON07}) \label{lemma:gentle} Consider a
density operator $\rho\in\mathcal{D}(\mathcal{H})$ and a measurement operator
$\Lambda$ where $0\leq\Lambda\leq I$. Suppose that the measurement operator
$\Lambda$ detects the state $\rho$ with high probability $\operatorname{Tr}%
\{\Lambda\rho\}\geq1-\varepsilon$, where $\varepsilon\in\lbrack0,1]$. Then%
\begin{equation}
\left\Vert \rho-\sqrt{\Lambda}\rho\sqrt{\Lambda}\right\Vert _{1}\leq
2\sqrt{\varepsilon}~.
\end{equation}

\end{lemma}

\begin{lemma}
\label{lemma:close} Consider two quantum states $\rho,\sigma\in\mathcal{D}%
(\mathcal{H})$ and a measurement operator $\Lambda$ where $0\leq\Lambda\leq
I$. Then we have
\begin{equation}
\operatorname{Tr}\{\Lambda\rho\}\geq\operatorname{Tr}\{\Lambda\sigma
\}-\left\Vert \rho-\sigma\right\Vert _{1}~.
\end{equation}
More generally, the same bound holds when $\rho$ and $\sigma$ are
subnormalized, i.e., $\operatorname{Tr}\{\rho\},\operatorname{Tr}%
\{\sigma\}\leq1$.
\end{lemma}

\bigskip

\textbf{Information spectrum.} The information spectrum approach
\cite{koga2013information,nagaoka2007information,hayashi2003generalCapacity,DR09}
gives one-shot bounds for operational tasks in quantum Shannon theory, with
very few assumptions made about the source or channel
\cite{nagaoka2007information,hayashi2003generalEntangle,hayashi2003generalCapacity,hayashi2007error}%
. What plays an important role in the information spectrum method is the
positive spectral projection of an operator. For a Hermitian operator $X$ with
spectral decomposition $X=\sum_{i}\lambda_{i}|i\rangle\langle i|$, the
associated \textit{positive spectral projection} is denoted and defined as
\begin{equation}
\{X\geq0\}\equiv\sum_{i\, :\, \lambda_{i}\geq0}|i\rangle\langle i|~.
\end{equation}

\bigskip

\textbf{Relative entropies and mutual informations.} For a state $\rho
\in\mathcal{D}(\mathcal{H})$ and a positive semi-definite operator $\sigma$,
the \textit{quantum R\'{e}nyi relative entropy} of order $\alpha$, where
$\alpha\in\lbrack0,1)\cup(1,+\infty)$ is defined as \cite{P86,TCR09}
\begin{equation}
D_{\alpha}(\rho\Vert\sigma)\equiv\frac{1}{\alpha-1}\log_{2}\operatorname{Tr}%
\{\rho^{\alpha}\sigma^{1-\alpha}\}~.
\end{equation}
If $\alpha>1$ and $\operatorname{supp}(\rho)\not \subseteq \operatorname{supp}%
(\sigma)$, it is set to $+\infty$. In the limit as $\alpha\rightarrow1$, the
above definition reduces to the \textit{quantum relative entropy} \cite{U62}
\begin{equation}
D(\rho\Vert\sigma)\equiv\operatorname{Tr}\{\rho\lbrack\log_{2}\rho-\log
_{2}\sigma]\}~,
\end{equation}
which is defined as above when $\operatorname{supp}(\rho)\subseteq
\operatorname{supp}(\sigma)$ and it is set to $+\infty$ otherwise. Using the
above definition, we can define the \textit{R\'{e}nyi mutual information} for
a bipartite state $\theta_{RB}$ as
\begin{equation}
I_{\alpha}(R;B)_{\theta}\equiv D_{\alpha}(\theta_{RB}\Vert\theta_{R}%
\otimes\theta_{B})~. \label{eq:Renyi-MI}%
\end{equation}
The $\varepsilon$\textit{-hypothesis testing relative entropy} for a state
$\rho$ and a positive semi-definite $\sigma$ is defined for $\varepsilon
\in\lbrack0,1]$ as \cite{BD10,wang2012one}
\begin{equation}
D_{H}^{\varepsilon}(\rho\Vert\sigma)\equiv-\log_{2}\inf_{\Lambda
}\{\operatorname{Tr}\{\Lambda\sigma\}:0\leq\Lambda\leq I\wedge
\operatorname{Tr}\{\Lambda\rho\}\geq1-\varepsilon\}~.
\end{equation}
Similarly, we define the $\varepsilon$\textit{-hypothesis testing mutual
information} of a bipartite state $\theta_{RB}$ as
\begin{equation}
I_{H}^{\varepsilon}(R;B)_{\theta}\equiv D_{H}^{\varepsilon}(\theta_{RB}%
\Vert\theta_{R}\otimes\theta_{B})~. \label{eq:hyp-test-MI}%
\end{equation}
Note that there are alternative definitions \cite{matthews2014finite} of
$\varepsilon$-hypothesis testing mutual information that involve an
optimization with respect to the marginal state on system $B$:%
\begin{equation}
\widetilde{I}_{H}^{\varepsilon}(R;B)_{\theta}\equiv\min_{\sigma_{B}}%
D_{H}^{\varepsilon}(\theta_{RB}\Vert\theta_{R}\otimes\sigma_{B}).
\end{equation}

The following proposition establishes an inequality relating hypothesis
testing relative entropy and the quantum R\'{e}nyi relative entropy, and it
represents a counterpart to \cite[Lemma~5]{CMW14}. We give its proof in the
appendix, where we also mention how \cite[Lemma~5]{CMW14} and the following
proposition lead to a transparent proof of the quantum Stein's lemma
\cite{hiai1991proper,ogawa2000strong}.

\begin{proposition}
\label{prop:ineq-hypo-renyi}Let $\rho$ be a density operator and let $\sigma$
be a positive semi-definite operator. Let $\alpha\in(0,1)$ and $\varepsilon
\in(0,1)$. Then the following inequality holds%
\begin{equation}
D_{H}^{\varepsilon}(\rho\Vert\sigma)\geq\frac{\alpha}{\alpha-1}\log
_{2}\!\left(  \frac{1}{\varepsilon}\right)  +D_{\alpha}(\rho\Vert\sigma).
\end{equation}

\end{proposition}

The hypothesis testing relative entropy has the following second-order
expansion \cite{tomamichel2013hierarchy,li2014second,datta2016second}:
\begin{equation}
D_{H}^{\varepsilon}(\rho^{\otimes n}\Vert\sigma^{\otimes n})=nD(\rho
\Vert\sigma)+\sqrt{nV(\rho\Vert\sigma)}\Phi^{-1}(\varepsilon)+O(\log n)~,
\label{eq:second-order-expansion}%
\end{equation}
where $V(\rho\Vert\sigma)=\operatorname{Tr}\{\rho\lbrack\log_{2}\rho-\log
_{2}\sigma]^{2}\}-\left[  D(\rho\Vert\sigma)\right]  ^{2}$ is the
\textit{quantum relative entropy variance} and the function $\Phi(a)$ is the
cumulative distribution function for a standard normal distribution:
\begin{equation}
\Phi(a)\equiv\frac{1}{\sqrt{2\pi}}\int_{-\infty}^{a}dx\ e^{-x^{2}/2}~.
\end{equation}

Let $\sigma$ be a quantum state now. The hypothesis testing relative entropy
is relevant for asymmetric hypothesis testing, in which the goal is to
minimize the error probability $\operatorname{Tr}\{\Lambda\sigma\}$ subject to
a constraint on the other kind of error probability $\operatorname{Tr}%
\{(I-\Lambda)\rho\}\leq\varepsilon$. We could also consider symmetric
hypothesis testing, in which the goal is to minimize both kinds of error
probabilities simultaneously. It is useful for us here to take the approach of
\cite{AM14}\ and consider general positive semi-definite operators $A$ and $B$
rather than states $\rho$ and $\sigma$. As in \cite{AM14}, we can define the
error \textquotedblleft probability\textquotedblright\ in identifying the
operators $A$ and $B$ as follows:%
\begin{align}
P_{e}^{\ast}(A,B)  &  \equiv\inf_{T\ :\ 0\leq T\leq I}\operatorname{Tr}%
\{(I-T)A\}+\operatorname{Tr}\{TB\}\label{eq:err-prob-sym-A-B-1}\\
&  =\operatorname{Tr}\{A\}-\sup_{T\ :\ 0\leq T\leq I}\operatorname{Tr}%
\{T(A-B)\}\\
&  =\operatorname{Tr}\{A\}-\operatorname{Tr}\{\{A-B\geq0\}(A-B)\}\\
&  =\frac{1}{2}\left(  \operatorname{Tr}\{A+B\}-\left\Vert A-B\right\Vert
_{1}\right)  . \label{eq:err-prob-sym-A-B-4}%
\end{align}
The following lemma allows for bounding $P_{e}^{\ast}(A,B)$ from above, and we
use it to establish bounds on the error exponent for entanglement-assisted communication.

\begin{lemma}
[\cite{ACMBMAV07}]\label{lemma:spectral-ineq} Let $A$ and $B$ be positive
semi-definite operators and $s\in\left[  0,1\right]  $. Then the following
inequality holds%
\begin{equation}
P_{e}^{\ast}(A,B)=\frac{1}{2}\left(  \operatorname{Tr}\{A+B\}-\left\Vert
A-B\right\Vert _{1}\right)  \leq\operatorname{Tr}\{A^{s}B^{1-s}\}.
\end{equation}

\end{lemma}

\noindent The above lemma was first proved in \cite{ACMBMAV07}, but the reader
should note that a much simpler proof due to N.~Ozawa is presented in
\cite[Proposition~1.1]{JOPS12} and \cite[Theorem~1]{A14}.

\bigskip\textbf{Weak typicality.} We will use results from the theory of weak
typicality in some of our achievability proofs (see, e.g.,
\cite{wilde2011classical,W17} for a review). Consider a density operator
$\rho_{A}$ with spectral decomposition: $\rho_{A}=\sum_{x}p_{X}(x)|x\rangle
\langle x|_{A}$. The weakly \textit{$\delta$-typical subspace} $T_{A^{n}%
}^{\rho,\delta}$ is defined as the span of all unit vectors $|x^{n}%
\rangle\equiv|x_{1}\rangle\otimes|x_{2}\rangle\otimes\cdots\otimes
|x_{n}\rangle$ such that the sample entropy $\overline{H}(x^{n})$ of their
classical label is close to the true entropy $H(X)=H(A)_{\rho}$ of the
distribution $p_{X}(x)$:%
\begin{equation}
T_{A^{n}}^{\rho,\delta}\equiv\text{span}\left\{  \left\vert x^{n}\right\rangle
:\left\vert \overline{H}(x^{n})-H(X)\right\vert \leq\delta\right\}  ,
\end{equation}
where $\overline{H}(x^{n})\equiv-\frac{1}{n}\log_{2}(p_{X^{n}}(x^{n}))$ and
$H(X)\equiv-\sum_{x}p_{X}(x)\log_{2}p_{X}(x)$. The $\delta$\textit{-typical
projector} $\Pi_{\rho,\delta}^{n}$\ onto the typical subspace of $\rho$ is
defined as%
\begin{equation}
\Pi_{A^{n}}^{\rho,\delta}\equiv\sum_{x^{n}\in T_{\delta}^{X^{n}}}|x^{n}%
\rangle\langle x^{n}|,
\end{equation}
where we have used the symbol $T_{X^{n}}^{\delta}$ to refer to the set of
$\delta$-typical sequences:%
\begin{equation}
T_{X^{n}}^{\delta}\equiv\left\{  x^{n}:\left\vert \overline{H}(x^{n}%
)-H(X)\right\vert \leq\delta\right\}  .
\end{equation}
Three important properties of the typical projector are as follows:%
\begin{align}
\operatorname{Tr}\{\Pi_{A^{n}}^{\rho,\delta}\rho^{\otimes n}\}  &
\geq1-\varepsilon,\label{ieq:typical-unit-prob}\\
\operatorname{Tr}\{\Pi_{A^{n}}^{\rho,\delta}\}  &  \leq2^{n\left[  H\left(
A\right)  +\delta\right]  },\\
2^{-n[H(A)+\delta]}\Pi_{A^{n}}^{\rho,\delta}\leq\Pi_{A^{n}}^{\rho,\delta}%
\rho^{\otimes n}\Pi_{A^{n}}^{\rho,\delta}  &  \leq2^{-n\left[  H(A)-\delta
\right]  }\Pi_{A^{n}}^{\rho,\delta}, \label{ieq:typical-equal-partition}%
\end{align}
where the first property holds for arbitrary $\varepsilon\in(0,1)$, $\delta
>0$, and sufficiently large $n$. We will also need the following `projector
trick' inequality \cite{giovannetti2012achieving,fawzi2012classical}:%
\begin{equation}
\Pi_{A^{n}}^{\rho,\delta}\leq2^{n[H(A)+\delta]}\rho_{A}^{\otimes n}~,
\label{ieq:project-trick}%
\end{equation}
which follows as a consequence of the leftmost inequality in
\eqref{ieq:typical-equal-partition} and the fact that $\Pi_{A^{n}}%
^{\rho,\delta}\rho^{\otimes n}\Pi_{A^{n}}^{\rho,\delta}=\sqrt{\rho^{\otimes
n}}\Pi_{A^{n}}^{\rho,\delta}\sqrt{\rho^{\otimes n}}\leq\rho^{\otimes n}$. A
final inequality we make use of is the following one%
\begin{equation}
\Pi_{A^{n}}^{\rho,\delta}\leq2^{-n\left[  H(A)-\delta\right]  /2}\left[
\rho^{\otimes n}\right]  ^{-1/2}, \label{eq:sqrt-root-proj-trick}%
\end{equation}
which is a consequence of sandwiching the rightmost inequality of
\eqref{ieq:typical-equal-partition} by $\left[  \rho^{\otimes n}\right]
^{-1/2}$, applying $\Pi_{A^{n}}^{\rho,\delta}\leq I^{\otimes n}$, and operator
monotonicity of the square root function.

\bigskip\textbf{Hayashi-Nagaoka operator inequality.} We repeatedly use the
following operator inequality from \cite{hayashi2003generalCapacity} when
analyzing error probability:

\begin{lemma}
\label{lemma:Hayashi-Nagaoka} Given operators $S$ and $T$ such that $0\leq
S\leq I$ and $T\geq0$, the following inequality holds for all $c>0$:
\begin{equation}
I-(S+T)^{-1/2}S(S+T)^{-1/2}\leq c_{\operatorname{I}}%
\,(I-S)+c_{\operatorname{II}}\,T~,
\end{equation}
where $c_{\operatorname{I}}\equiv1+c$ and $c_{\operatorname{II}}%
\equiv2+c+c^{-1}$.
\end{lemma}

\section{Entanglement-assisted point-to-point classical communication}

\label{sec:EA-point-to-point}

We begin by defining the information-processing task of point-to-point
entanglement-assisted classical communication, originally considered in
\cite{bennett2002entanglement,Hol01a}\ and studied further in
\cite{datta2013one-EA,matthews2014finite,DTW14}. Before communication begins,
the sender Alice and the receiver Bob share entanglement in whatever form they
wish, and we denote their shared state as $\Psi_{RA}$. Suppose Alice would
like to communicate some classical message $m$ from a set $\mathcal{M}%
\equiv\{1,\ldots,M\}$ over a quantum channel $\mathcal{N}_{A^{\prime
}\rightarrow B}$, where $M\in\mathbb{N}$ denotes the cardinality of the set
$\mathcal{M}$. An $(M,\varepsilon)$ entanglement-assisted classical code, for
$\varepsilon\in\lbrack0,1]$, consists of a collection $\{\mathcal{E}%
_{A\rightarrow A^{\prime}}^{m}\}_{m}$ of encoders and a decoding POVM
$\{\Lambda_{RB}^{m}\}_{m}$, such that the average error probability is bounded
from above by~$\varepsilon$:
\begin{equation}
\frac{1}{M}\sum_{m=1}^{M}\operatorname{Tr}\{(I-\Lambda_{RB}^{m})\mathcal{N}%
_{A^{\prime}\rightarrow B}(\mathcal{E}_{A\rightarrow A^{\prime}}^{m}(\Psi
_{RA}))\}\leq\varepsilon,
\end{equation}

For fixed $\varepsilon$, let $M^{\ast}(\mathcal{N},\varepsilon)$ denote the
largest $M$ for which there exists an $(M,\varepsilon)$ entanglement-assisted
classical communication code for the channel $\mathcal{N}$. Then we define the
$\varepsilon$-one-shot entanglement-assisted classical capacity as $\log
_{2}M^{\ast}(\mathcal{N},\varepsilon)$. We note that one could alternatively
consider maximum error probability when defining this capacity. The
entanglement-assisted capacity of a channel $\mathcal{N}$ is then defined as
\begin{equation}
C_{\operatorname{EA}}(\mathcal{N})\equiv\lim_{\varepsilon\rightarrow0}%
\liminf_{n\rightarrow\infty}\frac{1}{n}\log_{2}M^{\ast}(\mathcal{N}^{\otimes
n},\varepsilon).
\end{equation}

For fixed $M$, the one-shot entanglement-assisted error exponent $-\log
_{2}\varepsilon^{\ast}(\mathcal{N},M)$ is such that $\varepsilon^{\ast}$ is
equal to the smallest~$\varepsilon$ for which there exists an $(M,\varepsilon
)$ entanglement-assisted classical communication code. In the i.i.d.~setting,
the entanglement-assisted error exponent is defined for a fixed rate $R\geq0$
as
\begin{equation}
E_{\operatorname{EA}}(\mathcal{N},R)\equiv\limsup_{n\rightarrow\infty}\left[
-\frac{1}{n}\log_{2}\varepsilon^{\ast}(\mathcal{N}^{\otimes n},2^{nR})\right]
. \label{eq:EA-err-exp-def}%
\end{equation}

\subsection{One-shot position-based coding}

\label{subsec:EA-classical-one-shot}

We now review the method of position-based coding \cite{anshu2017one}, as
applied to point-to-point entanglement-assisted communication, and follow the
review by showing how the approach leads to a lower bound on the error
exponent for entanglement-assisted communication, a lower bound for one-shot
entanglement-assisted capacity, and a simple proof for a lower bound on the
second-order coding rate for entanglement-assisted communication. We note that
a lower bound for one-shot entanglement-assisted capacity using position-based
coding was already given in \cite{anshu2017one}, but the lower bound given
here leads to a lower bound on the entanglement-assisted second-order coding
rate that is optimal for covariant channels \cite{DTW14}.

The position-based entanglement-assisted communication protocol consists of
two steps, encoding and decoding, and we follow that discussion with an error
analysis of its performance.

\textbf{Encoding:} Before communication begins, Alice and Bob share the
following state:%
\begin{equation}
\theta_{RA}^{\otimes M}\equiv\theta_{R_{1}A_{1}}\otimes\cdots\otimes
\theta_{R_{M}A_{M}},
\end{equation}
where Alice possesses the $A$ systems and Bob has the $R$ systems.
%We think of this
%state as being composed of $M$ blocks, each block consisting of a single copy
%of $\theta_{RA}$.
To send message $m$, Alice simply sends the $m$th $A$ system through the
channel. So this leads to the following state for Bob:%
\begin{equation}
\rho_{R^{M}B}^{m}\equiv\theta_{R}^{\otimes m-1}\otimes\mathcal{N}%
_{A\rightarrow B}(\theta_{R_{m}A_{m}})\otimes\theta_{R}^{\otimes M-m}.
\end{equation}

\textbf{Decoding:} Define the following measurement:%
\begin{equation}
\Gamma_{R^{M}B}^{m}\equiv I_{R^{m-1}}\otimes T_{R_{m}B_{m}}\otimes I_{R^{M-m}%
},
\end{equation}
where $T_{R_{m}B_{m}}=T_{RB}$ is a \textquotedblleft test\textquotedblright%
\ or measurement operator satisfying $0\leq T_{RB}\leq I_{RB}$, which we will
specify later. For now, just think of it as corresponding to a measurement
that should distinguish well between $\mathcal{N}_{A\rightarrow B}(\theta
_{RA})$ and $\theta_{R}\otimes\mathcal{N}_{A\rightarrow B}(\theta_{A})$. This
is important for the following reason:\ If message $m$ is transmitted and the
test is performed on the $m$th $R$ system and the channel output system $B$,
then the probability of it accepting is%
\begin{equation}
\operatorname{Tr}\{\Gamma_{R^{M}B}^{m}\rho_{R^{M}B}^{m}\}=\operatorname{Tr}%
\{T_{RB}\mathcal{N}_{A\rightarrow B}(\theta_{RA})\}.
\end{equation}
If however the test is performed on the $m^{\prime}$th $R$ system and $B$,
where $m^{\prime}\neq m$, then the probability of it accepting is%
\begin{equation}
\operatorname{Tr}\{\Gamma_{R^{M}B}^{m^{\prime}}\rho_{R^{M}B}^{m}%
\}=\operatorname{Tr}\{T_{RB}[\theta_{R}\otimes\mathcal{N}_{A\rightarrow
B}(\theta_{A})]\}.
\end{equation}
We use these facts in the forthcoming error analysis.

We use a square-root measurement to form a decoding POVM for Bob as follows:%
\begin{equation}
\Lambda_{R^{M}B}^{m}\equiv\left(  \sum_{m^{\prime}=1}^{M}\Gamma_{R^{M}%
B}^{m^{\prime}}\right)  ^{-1/2}\Gamma_{R^{M}B}^{m}\left(  \sum_{m^{\prime}%
=1}^{M}\Gamma_{R^{M}B}^{m^{\prime}}\right)  ^{-1/2}.
\end{equation}
This is called the position-based decoder.

\textbf{Error analysis:} The error probability under this coding scheme is the
same for each message~$m$ (see, e.g., \cite{anshu2017one,wilde2017position})
and is as follows:%
\begin{equation}
p_{e}(m)\equiv\operatorname{Tr}\{(I_{R^{M}B}-\Lambda_{R^{M}B}^{m})\rho
_{R^{M}B}^{m}\}.
\end{equation}
Applying Lemma \ref{lemma:Hayashi-Nagaoka} with $S=\Gamma_{R^{M}B}^{m}$ and
$T=\sum_{m^{\prime}\neq m}\Gamma_{R^{M}B}^{m^{\prime}}$, we find that this
error probability can be bounded from above as%
\begin{align}
&  \operatorname{Tr}\{(I_{R^{M}B}-\Lambda_{R^{M}B}^{m})\rho_{R^{M}B}%
^{m}\}\nonumber\\
&  \leq c_{\operatorname{I}}\operatorname{Tr}\{(I_{R^{M}B}-\Gamma_{R^{M}B}%
^{m})\rho_{R^{M}B}^{m}\}+c_{\operatorname{II}}\sum_{m^{\prime}\neq
m}\operatorname{Tr}\{\Gamma_{R^{M}B}^{m^{\prime}}\rho_{R^{M}B}^{m}\}\\
&  =c_{\operatorname{I}}\operatorname{Tr}\{(I_{RB}-T_{RB})\mathcal{N}%
_{A\rightarrow B}(\theta_{RA})\}+c_{\operatorname{II}}\sum_{m^{\prime}\neq
m}\operatorname{Tr}\{T_{RB}\left[  \theta_{R}\otimes\mathcal{N}_{A\rightarrow
B}(\theta_{A})\right]  \}\\
&  =c_{\operatorname{I}}\operatorname{Tr}\{(I_{RB}-T_{RB})\mathcal{N}%
_{A\rightarrow B}(\theta_{RA})\}+c_{\operatorname{II}}(M-1)\operatorname{Tr}%
\{T_{RB}\left[  \theta_{R}\otimes\mathcal{N}_{A\rightarrow B}(\theta
_{A})\right]  \}.
\end{align}
The same bound applies for both the average and the maximum error probability,
due to the symmetric construction of the code.

Our bound for a test operator $T_{RB}$\ is thus as follows and highlights, as
in \cite{anshu2017one}, an important connection between quantum hypothesis
testing (i.e., the ability to distinguish the states $\mathcal{N}%
_{A\rightarrow B}(\theta_{RA})$ and $\theta_{R}\otimes\mathcal{N}%
_{A\rightarrow B}(\theta_{A})$) and entanglement-assisted communication:%
\begin{equation}
p_{e}(m)\leq c_{\operatorname{I}}\operatorname{Tr}\{(I_{RB}-T_{RB}%
)\mathcal{N}_{A\rightarrow B}(\theta_{RA})\}+c_{\operatorname{II}%
}(M-1)\operatorname{Tr}\{T_{RB}\left[  \theta_{R}\otimes\mathcal{N}%
_{A\rightarrow B}(\theta_{A})\right]  \}.
\label{eq:EA-classical-on-shot-bound}%
\end{equation}

\subsection{Lower bounds on one-shot and i.i.d.~entanglement-assisted error
exponents}

\label{subsec:EA-err-exp} We first prove a lower bound on the one-shot error
exponent, and then a lower bound for the entanglement-assisted error exponent
in the i.i.d.~case directly follows.

\begin{theorem}
For a quantum channel $\mathcal{N}_{A\rightarrow B}$, a lower bound on the
one-shot entanglement-assisted error exponent for fixed message size $M$ is as
follows:%
\begin{equation}
-\log_{2}\varepsilon^{\ast}(\mathcal{N},M)\geq\sup_{s\in\lbrack0,1]}\left(
1-s\right)  \left[  \sup_{\theta_{RA}}I_{s}(R;B)_{\mathcal{N}(\theta)}%
-\log_{2}M\right]  -2~, \label{eq:one-shot-err-exp-bound}%
\end{equation}
where $\theta_{RA}$ is a pure bipartite entangled state and $I_{s}%
(R;B)_{\mathcal{N}(\theta)}$ is the R\'{e}nyi mutual information defined in \eqref{eq:Renyi-MI}.
\end{theorem}

\begin{proof}
Following the position-based encoding and decoding procedure described in
Section~\ref{subsec:EA-classical-one-shot} and setting $c=1$ in
\eqref{eq:EA-classical-on-shot-bound}, the error probability for each message
can be bounded as%
\begin{align}
p_{e}(m)  &  =\operatorname{Tr}\{(I_{R^{M}B}-\Lambda_{R^{M}B}^{m})\rho
_{R^{M}B}^{m}\}\\
&  \leq4\left[  \operatorname{Tr}\{(I_{RB}-T_{RB})\mathcal{N}_{A\rightarrow
B}(\theta_{RA})\}+M\operatorname{Tr}\{T_{RB}\left[  \theta_{R}\otimes
\mathcal{N}_{A\rightarrow B}(\theta_{A})\right]  \}\right] \\
&  =4\left[  \operatorname{Tr}\{\mathcal{N}_{A\rightarrow B}(\theta
_{RA})\}-\operatorname{Tr}\{T_{RB}\left(  \mathcal{N}_{A\rightarrow B}%
(\theta_{RA})-M\left[  \theta_{R}\otimes\mathcal{N}_{A\rightarrow B}%
(\theta_{A})\right]  \right)  \}\right]  . \label{eq:EA-exp-first-bound}%
\end{align}
To minimize the term in the last line above, it is well known that one should
take the test operator $T_{RB}$ as follows:%
\begin{equation}
T_{RB}=\left\{  \mathcal{N}_{A\rightarrow B}(\theta_{RA})-M\left[  \theta
_{R}\otimes\mathcal{N}_{A\rightarrow B}(\theta_{A})\right]  \geq0\right\}  .
\end{equation}
The statement for quantum states is due to \cite{H69,japan1973holevo,Hel76}%
\ and the extension (relevant for us) to the more general case of positive
semi-definite operators appears in \cite[Eq.~(22)]{AM14} (see also
\eqref{eq:err-prob-sym-A-B-1}--\eqref{eq:err-prob-sym-A-B-4}). This then leads
to the following upper bound on the error probability:%
\begin{align}
p_{e}(m)  &  \leq4\left[  \operatorname{Tr}\{\mathcal{N}_{A\rightarrow
B}(\theta_{RA})\}-\operatorname{Tr}\{T_{RB}\left(  \mathcal{N}_{A\rightarrow
B}(\theta_{RA})-M\left[  \theta_{R}\otimes\mathcal{N}_{A\rightarrow B}%
(\theta_{A})\right]  \right)  \}\right] \\
&  =2\Big[\operatorname{Tr}\{\mathcal{N}_{A\rightarrow B}(\theta
_{RA})+M\left[  \theta_{R}\otimes\mathcal{N}_{A\rightarrow B}(\theta
_{A})\right]  \}\nonumber\\
&  \qquad\qquad-\left\Vert \mathcal{N}_{A\rightarrow B}(\theta_{RA})-M\left[
\theta_{R}\otimes\mathcal{N}_{A\rightarrow B}(\theta_{A})\right]  \right\Vert
_{1}\Big]\\
&  \leq4\operatorname{Tr}\{\mathcal{N}_{A\rightarrow B}(\theta_{RA}%
)^{s}\left(  M\theta_{R}\otimes\mathcal{N}_{A\rightarrow B}(\theta
_{A})\right)  ^{1-s}\}\\
&  =4M^{1-s}\operatorname{Tr}\{\mathcal{N}_{A\rightarrow B}(\theta_{RA}%
)^{s}\left[  \theta_{R}\otimes\mathcal{N}_{A\rightarrow B}(\theta_{A})\right]
^{1-s}\}\\
&  =4\left(  2^{-(1-s)\left[  I_{s}(R;B)_{\mathcal{N}(\theta)}-\log
_{2}M\right]  }\right)  .
\end{align}
The first equality is standard, using the relation of the positive part of an
operator to its modulus (see, e.g., \cite[Eq.~(23)]{AM14}). The second
inequality is a consequence of \cite[Theorem~1]{ACMBMAV07}, recalled as
Lemma~\ref{lemma:spectral-ineq}\ in Section~\ref{sec:preliminary}, and holds
for all $s\in\lbrack0,1]$ (see \cite[Proposition~1.1]{JOPS12} and
\cite[Theorem~1]{A14} for a simpler proof of \cite[Theorem~1]{ACMBMAV07} due
to N.~Ozawa). The last equality follows from the definition of R\'{e}nyi
mutual information in \eqref{eq:Renyi-MI}. Since this bound holds for an
arbitrary $s\in\lbrack0,1]$ and an arbitrary input state $\theta_{RA}$, we can
conclude the following bound:
\begin{equation}
p_{e}(m)\leq4\left(  2^{-\sup_{s\in\lbrack0,1]}(1-s)\left[  \sup_{\theta_{RA}%
}I_{s}(R;B)_{\mathcal{N}(\theta)}-\log_{2}M\right]  }\right)  ~.
\label{eq:final-1-shot-err-exp}%
\end{equation}
Note that it suffices to take $\theta_{RA}$ as a pure bipartite state, due to
the ability to purify a mixed $\theta_{RA}$\ and the data-processing
inequality for $I_{s}(R;B)_{\mathcal{N}(\theta)}$, holding for all
$s\in\lbrack0,1]$ \cite{P86}. Finally taking a negative binary logarithm of
both sides of \eqref{eq:final-1-shot-err-exp} gives \eqref{eq:one-shot-err-exp-bound}.
\end{proof}

We remark that the above proof bears some similarities to that given in
\cite{hayashi2007error} (one can find a related result in the later work
\cite[Lemma~3.1]{MD09}). One of the results in \cite{hayashi2007error}
concerns a bound on the error exponent for classical communication over
classical-input quantum-output channels. The fundamental tool used in the
proof of this result in \cite{hayashi2007error} is
Lemma~\ref{lemma:spectral-ineq}, attributed above to \cite{ACMBMAV07}. Our
proof above clearly follows the same approach.\bigskip

Applying the above result in the i.i.d.~case for a memoryless channel
$\mathcal{N}_{A\rightarrow B}^{\otimes n}$ leads to the following:

\begin{proposition}
For a quantum channel $\mathcal{N}_{A\rightarrow B}$, a lower bound on the
entanglement-assisted error exponent $E_{\operatorname{EA}}(\mathcal{N},R)$
(defined in \eqref{eq:EA-err-exp-def}) for fixed rate $R\geq0$ is as follows:%
\begin{equation}
E_{\operatorname{EA}}(\mathcal{N},R)\geq\sup_{s\in\lbrack0,1]}\left(
1-s\right)  \left[  \sup_{\theta_{RA}}I_{s}(R;B)_{\mathcal{N}(\theta
)}-R\right]  ~, \label{eq:EA-err-exp-lower}%
\end{equation}
where $\theta_{RA}$ is a pure bipartite entangled state and $I_{s}%
(R;B)_{\mathcal{N}(\theta)}$ is the R\'{e}nyi mutual information defined in \eqref{eq:Renyi-MI}.
\end{proposition}

\begin{proof}
A proof follows by plugging in the memoryless channel $\mathcal{N}%
_{A\rightarrow B}^{\otimes n}$ into \eqref{eq:one-shot-err-exp-bound}, setting
the number of messages to be $M=2^{nR}$, and considering that
\begin{align}
\sup_{\theta_{R^{n}A^{n}}^{(n)}}I_{s}(R^{n};B^{n})_{\mathcal{N}^{\otimes
n}(\theta^{(n)})}  &  \geq\sup_{\theta_{RA}^{\otimes n}}I_{s}(R^{n}%
;B^{n})_{\left[  \mathcal{N}(\theta)\right]  ^{\otimes n}}\\
&  =n\sup_{\theta_{RA}}I_{s}(R;B)_{\mathcal{N}(\theta)},
\label{eq:additive-Renyi}%
\end{align}
leading to the following bound:%
\begin{equation}
-\frac{1}{n}\log_{2}\varepsilon^{\ast}(\mathcal{N}^{\otimes n},2^{nR})\geq
\sup_{s\in\lbrack0,1]}\left(  1-s\right)  \left[  \sup_{\theta_{RA}}%
I_{s}(R;B)_{\mathcal{N}(\theta)}-R\right]  -\frac{2}{n}.
\end{equation}
The equality in \eqref{eq:additive-Renyi} follows from the additivity of the
R\'{e}nyi mutual information for tensor-power states. Taking the large $n$
limit then gives \eqref{eq:EA-err-exp-lower}. Alternatively, plugging the
memoryless channel $\mathcal{N}_{A\rightarrow B}^{\otimes n}$ in to
\eqref{eq:final-1-shot-err-exp}, we find that the bound on the error
probability becomes%
\begin{equation}
p_{e}(m)\leq4\left(  2^{-(1-s)n\left[  I_{s}(R;B)_{\mathcal{N}(\theta
)}-R\right]  }\right)  ,
\end{equation}
holding for all $s\in\lbrack0,1]$ and states $\theta_{RA}$. After taking a
negative logarithm, normalizing by $n$, and taking the limit as $n\rightarrow
\infty$, we arrive at \eqref{eq:EA-err-exp-lower}.
\end{proof}

\subsection{Lower bounds on one-shot entanglement-assisted capacity and
entanglement-assisted second-order coding rate}

By using position-based coding, here we establish a lower bound on the
one-shot entanglement-assisted capacity. Note that a similar lower bound was
established in \cite{anshu2017one}, but the theorem below allows for an
additional parameter $\eta\in(0,\varepsilon)$, which is helpful for giving a
lower bound on the entanglement-assisted second-order coding rate.

\begin{theorem}
Given a quantum channel $\mathcal{N}_{A\rightarrow B}$ and fixed
$\varepsilon\in(0,1)$, the $\varepsilon$-one-shot entanglement-assisted
capacity of $\mathcal{N}_{A\rightarrow B}$ is bounded as
\begin{equation}
\log_{2}M^{\ast}(\mathcal{N},\varepsilon)\geq\max_{\theta_{RA}}I_{H}%
^{\varepsilon-\eta}(R;B)_{\mathcal{N}(\theta)}-\log_{2}(4\varepsilon/\eta
^{2}), \label{eq:one-shot-ea-lower}%
\end{equation}
where $\eta\in(0,\varepsilon)$ and the hypothesis testing mutual information
is defined in \eqref{eq:hyp-test-MI}.
\end{theorem}

\begin{proof}
The idea is to use the same coding scheme described in
Section~\ref{subsec:EA-classical-one-shot} and take the test operator $T_{RB}$
in Bob's decoder to be $\Upsilon_{RB}^{\ast}$, where $\Upsilon_{RB}^{\ast}$ is
the optimal measurement operator for $I_{H}^{\varepsilon-\eta}%
(R;B)_{\mathcal{N}(\theta)}$, with $\eta\in(0,\varepsilon)$. Then, starting
from the upper bound on the error probability in
\eqref{eq:EA-classical-on-shot-bound}, the error analysis reduces to%
\begin{align}
&  \operatorname{Tr}\{(I_{R^{M}B}-\Lambda_{R^{M}B})\rho_{R^{M}B}%
^{m}\}\nonumber\\
&  \leq c_{\operatorname{I}}\operatorname{Tr}\{(I_{RB}-\Upsilon_{RB}^{\ast
})\left[  \mathcal{N}_{A\rightarrow B}(\theta_{RA})\right]
\}+c_{\operatorname{II}}M\operatorname{Tr}\{\Upsilon_{RB}^{\ast}\left[
\theta_{R}\otimes\mathcal{N}_{A\rightarrow B}(\theta_{A})\right]  \}\\
&  \leq c_{\operatorname{I}}\left(  \varepsilon-\eta\right)
+c_{\operatorname{II}}M2^{-I_{H}^{\varepsilon-\eta}(R;B)_{\mathcal{N}(\theta
)}}.
\end{align}
The second inequality follows from the definition of quantum hypothesis
testing relative entropy, which gives that%
\begin{align}
\operatorname{Tr}\{\Upsilon_{RB}^{\ast}\left[  \mathcal{N}_{A\rightarrow
B}(\theta_{RA})\right]  \}  &  \geq1-(\varepsilon-\eta),\\
\operatorname{Tr}\{\Upsilon_{RB}^{\ast}\left[  \theta_{R}\otimes
\mathcal{N}_{A\rightarrow B}(\theta_{A})\right]  \}  &  =2^{-I_{H}%
^{\varepsilon-\eta}(R;B)_{\mathcal{N}(\theta)}}.
\end{align}
To make the error $p_{e}(m)\leq\varepsilon$, we set $c=\eta/(2\varepsilon
-\eta)$ for $\eta\in(0,\varepsilon)$,\ and this leads to
\begin{equation}
\log_{2}M=I_{H}^{\varepsilon-\eta}(R;B)_{\mathcal{N}(\theta)}-\log
_{2}(4\varepsilon/\eta^{2}).
\end{equation}
The inequality in the theorem follows after maximizing $I_{H}^{\varepsilon
-\eta}(R;B)_{\mathcal{N}(\theta)}$ with respect to all input states
$\theta_{RA}$.
\end{proof}

\bigskip\textbf{Comparison to upper bound.} The authors of
\cite{matthews2014finite} established the following upper bound on one-shot
entanglement-assisted capacity:%
\begin{align}
\max_{\theta_{RA}} I_{H}^{\varepsilon}(R;B)_{\mathcal{N}(\theta)}  &  \geq
\max_{\theta_{RA}}\min_{\sigma_{B}}D_{H}^{\varepsilon}(\mathcal{N}%
_{A\rightarrow B}(\theta_{RA})\Vert\theta_{R}\otimes\sigma_{B})\\
&  \geq\log_{2}M^{\ast}(\mathcal{N},\varepsilon).
\end{align}
Thus, there is a sense in which the upper bound from \cite{matthews2014finite}
is close to the lower bound in \eqref{eq:one-shot-ea-lower}. In particular, we
could pick $\eta=\delta\varepsilon$ for any constant $\delta\in(0,1)$, and the
lower bound becomes $\max_{\theta_{RA}}I_{H}^{\varepsilon(1-\delta
)}(R;B)_{\mathcal{N}(\theta)}-\log_{2}(4/\varepsilon\delta^{2})$. Thus the
information term $\max_{\theta_{RA}}I_{H}^{\varepsilon(1-\delta)}%
(R;B)_{\mathcal{N}(\theta)}$ can become arbitrarily close to $I_{H}%
^{\varepsilon}(R;B)_{\mathcal{N}(\theta)}$ by picking $\delta$ smaller, but at
the cost of the term $-\log_{2}(4/\varepsilon\delta^{2})$ becoming more
negative with decreasing $\delta$.

\bigskip\textbf{Lower bound on second-order coding rate.} To get a lower bound
on the entanglement-assisted second-order coding rate for an i.i.d.~channel
$\mathcal{N}^{\otimes n}$, evaluate the formula $I_{H}^{\varepsilon-\eta
}(R;B)_{\mathcal{N}(\theta)}$ for an i.i.d.~state $\mathcal{N}(\theta
)^{\otimes n}$, pick $\eta=1/\sqrt{n}$ and $n$ large enough such that
$\varepsilon-\eta>0$, and use the second-order expansions for $D_{H}%
^{\varepsilon}$ in \eqref{eq:second-order-expansion}. We then recover one of
the main results of \cite{DTW14}:%
\begin{equation}
\log_{2}M^{\ast}(\mathcal{N}^{\otimes n},\varepsilon)\geq nI(R;B)_{\mathcal{N}%
(\theta)}+\sqrt{nV(R;B)_{\mathcal{N}(\theta)}}\Phi^{-1}(\varepsilon)+O(\log
n).
\end{equation}
Interestingly, this is achievable for maximal error in addition to average
error due to the above analysis. Additionally, it does seem that this approach
for arriving at a lower bound on the entanglement-assisted second-order coding
rate is much simpler than the previous approach developed in \cite{DTW14}.

\subsection{Alternative proof of an upper bound on one-shot
entanglement-assisted capacity}

In this section, we provide a proof for an upper bound on the one-shot
entanglement-assisted classical capacity, which is arguably simpler than the
approach taken in \cite{matthews2014finite}. A proof along these lines was
found recently and independently in \cite{AJW18}.

Before doing so, we recall the definition of generalized divergence
$\mathbf{D}(\rho\Vert\sigma)$\ \cite{SW12}\ of two states $\rho$ and $\sigma$
as any function from two density operators to the reals that is monotone under
the action of a quantum channel $\mathcal{N}$, in the sense that%
\begin{equation}
\mathbf{D}(\rho\Vert\sigma)\geq\mathbf{D}(\mathcal{N}(\rho)\Vert
\mathcal{N}(\sigma)).
\end{equation}
From this, we can define the generalized mutual information of a bipartite
state $\rho_{AB}$ as%
\begin{equation}
I_{\mathbf{D}}(A;B)\equiv\inf_{\sigma_{B}}\mathbf{D}(\rho_{AB}\Vert\rho
_{A}\otimes\sigma_{B}),
\end{equation}
where the optimization is with respect to a density operator $\sigma_{B}$
acting on system $B$. We then have the following lemma:

\begin{lemma}
\label{lem:marginal-prod} Let $\rho_{ABC}$ be such that the marginal state
$\rho_{AC}$ is product (i.e., $\rho_{AC}=\rho_{A}\otimes\rho_{C}$). Then%
\begin{equation}
I_{\mathbf{D}}(A;BC)_{\rho}\leq I_{\mathbf{D}}(AC;B)_{\rho}.
\end{equation}

\end{lemma}

\begin{proof}
This follows because%
\begin{align}
I_{\mathbf{D}}(A;BC)_{\rho}  &  =\inf_{\sigma_{BC}}\mathbf{D}(\rho_{ABC}%
\Vert\rho_{A}\otimes\sigma_{BC})\\
&  \leq\inf_{\sigma_{B}}\mathbf{D}(\rho_{ABC}\Vert\rho_{A}\otimes\sigma
_{B}\otimes\rho_{C})\\
&  =\inf_{\sigma_{B}}\mathbf{D}(\rho_{ABC}\Vert\rho_{AC}\otimes\sigma_{B})\\
&  =I_{\mathbf{D}}(AC;B)_{\rho}.
\end{align}
This concludes the proof.
\end{proof}

We now apply Lemma~\ref{lem:marginal-prod} in the context of
entanglement-assisted communication, to establish an alternate proof of the
following upper bound (again emphasizing that a proof along these lines was
found recently and independently in \cite{AJW18}):

\begin{proposition}
[\cite{matthews2014finite}]\label{prop-eac:one-shot-bound-meta}Let
$\mathcal{N}_{A\rightarrow B}$ be a quantum channel. For an $(\left\vert
M\right\vert ,\varepsilon)$ entanglement-assisted classical communication
protocol over the channel $\mathcal{N}_{A\rightarrow B}$, the following bound
holds%
\begin{equation}
\log_{2}\left\vert M\right\vert \leq I^{\varepsilon}(\mathcal{N}),
\end{equation}
where the $\varepsilon$-mutual information of $\mathcal{N}_{A\rightarrow B}$
is defined as%
\begin{equation}
I^{\varepsilon}(\mathcal{N})\equiv\max_{\psi_{RA}}\min_{\sigma_{B}}%
D_{H}^{\varepsilon}(\mathcal{N}_{A\rightarrow B}(\psi_{RA})\Vert\psi
_{R}\otimes\sigma_{B}),
\end{equation}
with $\psi_{RA}$ a pure bipartite state such that system $R$ isomorphic to
system $A$.
\end{proposition}

\begin{proof}
An entanglement-assisted classical communication protocol begins with the
sender preparing the maximally classically correlated state $\overline{\Phi
}_{MM^{\prime}}$, defined as%
\begin{equation}
\overline{\Phi}_{MM^{\prime}}\equiv\frac{1}{\left\vert M\right\vert }%
\sum_{m=1}^{\left\vert M\right\vert }|m\rangle\langle m|_{M}\otimes
|m\rangle\langle m|_{M^{\prime}}.
\end{equation}
Also, the sender and receiver share an arbitrary entangled state
$\Psi_{A^{\prime}B^{\prime}}$ before communication begins. The sender then
performs an encoding channel $\mathcal{E}_{M^{\prime}A^{\prime}\rightarrow A}$
on systems $M^{\prime}$ and $A^{\prime}$, and the resulting state is%
\begin{equation}
\mathcal{E}_{M^{\prime}A^{\prime}\rightarrow A}(\overline{\Phi}_{MM^{\prime}%
}\otimes\Psi_{A^{\prime}B^{\prime}})=\frac{1}{\left\vert M\right\vert }%
\sum_{m=1}^{\left\vert M\right\vert }|m\rangle\langle m|_{M}\otimes
\mathcal{E}_{M^{\prime}A^{\prime}\rightarrow A}(|m\rangle\langle
m|_{M^{\prime}}\otimes\Psi_{A^{\prime}B^{\prime}}).
\end{equation}
Defining the quantum channels $\mathcal{E}_{A^{\prime}\rightarrow A}^{m}$ by
$\mathcal{E}_{A^{\prime}\rightarrow A}^{m}(\tau_{A^{\prime}})\equiv
\mathcal{E}_{M^{\prime}A^{\prime}\rightarrow A}(|m\rangle\langle
m|_{M^{\prime}}\otimes\tau_{A^{\prime}})$, we can write the above state as%
\begin{equation}
\mathcal{E}_{M^{\prime}A^{\prime}\rightarrow A}(\overline{\Phi}_{MM^{\prime}%
}\otimes\Psi_{A^{\prime}B^{\prime}})=\frac{1}{\left\vert M\right\vert }%
\sum_{m=1}^{\left\vert M\right\vert }|m\rangle\langle m|_{M}\otimes
\mathcal{E}_{A^{\prime}\rightarrow A}^{m}(\Psi_{A^{\prime}B^{\prime}}).
\end{equation}
The sender transmits the $A$ system through the channel $\mathcal{N}%
_{A\rightarrow B}$, leading to%
\begin{equation}
(\mathcal{N}_{A\rightarrow B}\circ\mathcal{E}_{M^{\prime}A^{\prime}\rightarrow
A})(\overline{\Phi}_{MM^{\prime}}\otimes\Psi_{A^{\prime}B^{\prime}})=\frac
{1}{\left\vert M\right\vert }\sum_{m=1}^{\left\vert M\right\vert }%
|m\rangle\langle m|_{M}\otimes(\mathcal{N}_{A\rightarrow B}\circ
\mathcal{E}_{A^{\prime}\rightarrow A}^{m})(\Psi_{A^{\prime}B^{\prime}%
}).\label{eq-eac:state-after-channel}%
\end{equation}
The receiver's goal is then to determine which message $m$ was transmitted. To
do so, he performs a quantum-to-classical or measurement channel
$\mathcal{D}_{BB^{\prime}\rightarrow\hat{M}}$, defined by%
\begin{equation}
\mathcal{D}_{BB^{\prime}\rightarrow\hat{M}}(\tau_{BB^{\prime}}):=\sum
_{m}\operatorname{Tr}[\Lambda_{BB^{\prime}}^{m}\tau_{BB^{\prime}}%
]|m\rangle\langle m|_{\hat{M}},
\end{equation}
for a POVM\ $\{\Lambda_{B}^{m}\}_{m=1}^{\left\vert M\right\vert }$, and the
state becomes%
\begin{align}
\omega_{M\hat{M}} &  :=(\mathcal{D}_{BB^{\prime}\rightarrow\hat{M}}%
\circ\mathcal{N}_{A\rightarrow B}\circ\mathcal{E}_{M^{\prime}A^{\prime
}\rightarrow A})(\overline{\Phi}_{MM^{\prime}}\otimes\Psi_{A^{\prime}%
B^{\prime}})\label{eq-eac:end-EAC-state}\\
&  =\frac{1}{\left\vert M\right\vert }\sum_{m,m^{\prime}=1}^{\left\vert
M\right\vert }|m\rangle\langle m|_{M}\otimes\operatorname{Tr}[\Lambda
_{BB^{\prime}}^{m^{\prime}}\mathcal{N}_{A\rightarrow B}(\mathcal{E}%
_{A^{\prime}\rightarrow A}^{m}(\Psi_{A^{\prime}B^{\prime}}))]|m^{\prime
}\rangle\langle m^{\prime}|_{\hat{M}}.
\end{align}
The protocol is an $(\left\vert M\right\vert ,\varepsilon)$ protocol by
definition if the following condition holds%
\begin{equation}
1-\frac{1}{\left\vert M\right\vert }\sum_{m}p(m|m)\leq\varepsilon
,\label{eq-eac:error-cond}%
\end{equation}
where%
\begin{equation}
p(m^{\prime}|m):=\operatorname{Tr}[\Lambda_{BB^{\prime}}^{m^{\prime}%
}\mathcal{N}_{A\rightarrow B}(\mathcal{E}_{A^{\prime}\rightarrow A}^{m}%
(\Psi_{A^{\prime}B^{\prime}}))].
\end{equation}
The following equality holds by direct calculation:%
\begin{equation}
1-\frac{1}{\left\vert M\right\vert }\sum_{m}p(m|m)=1-\operatorname{Tr}%
[\Pi_{M\hat{M}}\omega_{M\hat{M}}],\label{eq:comp-test}%
\end{equation}
where the comparator test or projection $\Pi_{M\hat{M}}$ is defined as%
\begin{equation}
\Pi_{M\hat{M}}:=\sum_{m=1}^{\left\vert M\right\vert }|m\rangle\langle
m|_{M}\otimes|m\rangle\langle m|_{M^{\prime}}.
\end{equation}
Now, let us apply the error condition in \eqref{eq-eac:error-cond}, the
equality in \eqref{eq:comp-test}, and the definition of hypothesis testing
relative entropy to conclude that%
\begin{equation}
\log_{2}\left\vert M\right\vert \leq\widetilde{I}_{H}^{\varepsilon}(M;\hat
{M})_{\omega},
\end{equation}
where $\omega_{M\hat{M}}$ is defined in \eqref{eq-eac:end-EAC-state}. From
data processing under the action of the decoding channel $\mathcal{D}%
_{BB^{\prime}\rightarrow\hat{M}}$, we find that%
\begin{equation}
\widetilde{I}_{H}^{\varepsilon}(M;\hat{M})_{\omega}\leq\widetilde{I}%
_{H}^{\varepsilon}(M;BB^{\prime})_{\theta},
\end{equation}
where the state $\theta_{MBB^{\prime}}$ is the same as that in
\eqref{eq-eac:state-after-channel}:%
\begin{equation}
\theta_{MBB^{\prime}}:=(\mathcal{N}_{A\rightarrow B}\circ\mathcal{E}%
_{M^{\prime}A^{\prime}\rightarrow A})(\overline{\Phi}_{MM^{\prime}}\otimes
\Psi_{A^{\prime}B^{\prime}}).
\end{equation}
Observe that the reduced state $\theta_{MB^{\prime}}$ is a product state
because the channel $\mathcal{N}_{A\rightarrow B}$ and encoding $\mathcal{E}%
_{M^{\prime}A^{\prime}\rightarrow A}$ are trace preserving:%
\begin{align}
\theta_{MB^{\prime}} &  =\operatorname{Tr}_{B}[(\mathcal{N}_{A\rightarrow
B}\circ\mathcal{E}_{M^{\prime}A^{\prime}\rightarrow A})(\overline{\Phi
}_{MM^{\prime}}\otimes\Psi_{A^{\prime}B^{\prime}})]\\
&  =\operatorname{Tr}_{M^{\prime}A^{\prime}}[\overline{\Phi}_{MM^{\prime}%
}\otimes\Psi_{A^{\prime}B^{\prime}}]\\
&  =\overline{\Phi}_{M}\otimes\Psi_{B^{\prime}}\\
&  =\theta_{M}\otimes\theta_{B^{\prime}}%
.\label{eq-eac:initial-reduced-state-product}%
\end{align}
Thus, from Lemma~\ref{lem:marginal-prod}, we have that%
\begin{equation}
\widetilde{I}_{H}^{\varepsilon}(M;BB^{\prime})_{\theta}\leq\widetilde{I}%
_{H}^{\varepsilon}(MB^{\prime};B)_{\theta}.
\end{equation}
Consider that the state $\theta_{MB^{\prime}B}$ has the form $\mathcal{N}%
_{A\rightarrow B}(\rho_{SA})$ for some mixed state $\rho_{SA}$, by identifying
$S$ with $MB^{\prime}$ and $\rho_{SA}$ with $\mathcal{E}_{M^{\prime}A^{\prime
}\rightarrow A}(\overline{\Phi}_{MM^{\prime}}\otimes\Psi_{A^{\prime}B^{\prime
}})$. So we find that%
\begin{equation}
\widetilde{I}_{H}^{\varepsilon}(MB^{\prime};B)_{\theta}\leq\max_{\rho_{SA}%
}\widetilde{I}_{H}^{\varepsilon}(S;B)_{\xi},
\end{equation}
where%
\begin{equation}
\xi_{SB}:=\mathcal{N}_{A\rightarrow B}(\rho_{SA}).
\end{equation}
Now employing the fact that $\rho_{SA}$ can be purified to $\psi_{S^{\prime
}SA}$, the data processing inequality for the $\varepsilon$-mutual
information, and the Schmidt decomposition theorem with respect to the
bipartite cut $S^{\prime}S|A$ to conclude that the Schmidt rank of
$\psi_{S^{\prime}SA}$ is no larger than $\left\vert A\right\vert $, we
conclude the statement of the proposition.
\end{proof}

%\subsection{Note on unassisted communication}
%The whole development above recovers the case of unassisted classical
%communication. We could take the system $R$ of the state $\theta_{RA}$ to be
%classical, in which case the shared state $\theta_{RA}$ could be realized by
%shared randomness only, using a preparation channel $|x\rangle\langle
%x|_{A^{\prime}}\rightarrow\rho_{A}^{x}$ in the following way:%
%\begin{equation}
%\theta_{RA^{\prime}}=\sum_{x}p_{X}(x)|x\rangle\langle x|_{R}\otimes
%|x\rangle\langle x|_{A^{\prime}}\rightarrow\theta_{RA}=\sum_{x}p_{X}%
%(x)|x\rangle\langle x|_{R}\otimes\rho_{A}^{x}.
%\end{equation}
%Then all of the above achievability statements hold. However, in this case the
%shared randomness can be derandomized, so that it is ultimately not necessary
%(for this, we would really need to consider average error probability rather
%than maximal error probability and then derandomize).
%%%%%%%%%%%%%%%%%%%%%%%%%%%%

\section{Classical communication over quantum multiple-access channels}

\label{sec:EA-MAC}

We now establish a link between communication over a quantum multiple-access
channel and multiple quantum hypothesis testing. One advantage of this
development is the reduction of the communication problem to a hypothesis
testing problem, which is perhaps simpler to state and could also be
considered a more fundamental problem than the communication problem. Later,
in Section~\ref{sec:disucssion}, we discuss the relation of the quantum
simultaneous decoding conjecture from \cite{fawzi2012classical,Wil11}\ to open
questions in multiple quantum hypothesis testing from \cite{AM14,BHOS14} (here
we note that the solution of the multiple Chernoff distance conjecture from
\cite{li2016discriminating}\ does not evidently allow for the solution of the
quantum simultaneous decoding conjecture). We point the reader to
\cite{D11,DF13} for further discussions and variations of the quantum
simultaneous decoding conjecture.

We begin by considering the case of entanglement-assisted communication and
later consider unassisted communication. We first define the
information-processing task of entanglement-assisted classical communication
over quantum multiple-access channels (see also
\cite{hsieh2008entanglement,Xu2013}). Consider the scenario in which two
senders Alice and Bob would like to transmit classical messages to a receiver
Charlie over a two-sender single-receiver quantum multiple-access channel
$\mathcal{N}_{AB\rightarrow C}$. Alice and Bob choose their messages from
message sets $\mathcal{L}$ and $\mathcal{M}$. The cardinality of the sets
$\mathcal{L}$ and $\mathcal{M}$ are denoted as $L$ and $M$, respectively.
Suppose that Alice and Bob each share an arbitrary entangled state with
Charlie before communication begins. Let $\Phi_{RA^{\prime}}$ denote the state
shared between Charlie and Alice, and let $\Psi_{SB^{\prime}}$ denote the
state shared between Charlie and Bob.

Let $L,M\in\mathbb{N}$ and $\varepsilon\in\lbrack0,1]$. An $(L,M,\varepsilon)$
entanglement-assisted multiple-access classical communication code consists of
a set $\{\mathcal{E}_{A^{\prime}\rightarrow A}^{l},\mathcal{F}_{B^{\prime
}\rightarrow B}^{m},\Lambda_{RSC}^{l,m}\}_{l,m}$ of encoders and a decoding
POVM, such that the average error probability is bounded from above by
$\varepsilon$:
\begin{equation}
\frac{1}{LM}\sum_{l=1}^{L}\sum_{m=1}^{M}p_{e}(l,m)\leq\varepsilon~,
\end{equation}
where the error probability for each message pair is given by
\begin{equation}
p_{e}(l,m)\equiv\operatorname{Tr}\{(I_{RSC}-\Lambda_{RSC}^{l,m})\mathcal{N}%
_{AB\rightarrow C}(\mathcal{E}_{A^{\prime}\rightarrow A}^{l}(\Phi_{RA^{\prime
}})\otimes\mathcal{F}_{B^{\prime}\rightarrow B}^{m}(\Psi_{SB^{\prime}}))\}~.
\end{equation}
Figure~\ref{fig:EA-MAC-task}\ depicts the coding task.

\begin{figure}[ptb]
\begin{center}
\includegraphics[
width=3.4809in
]{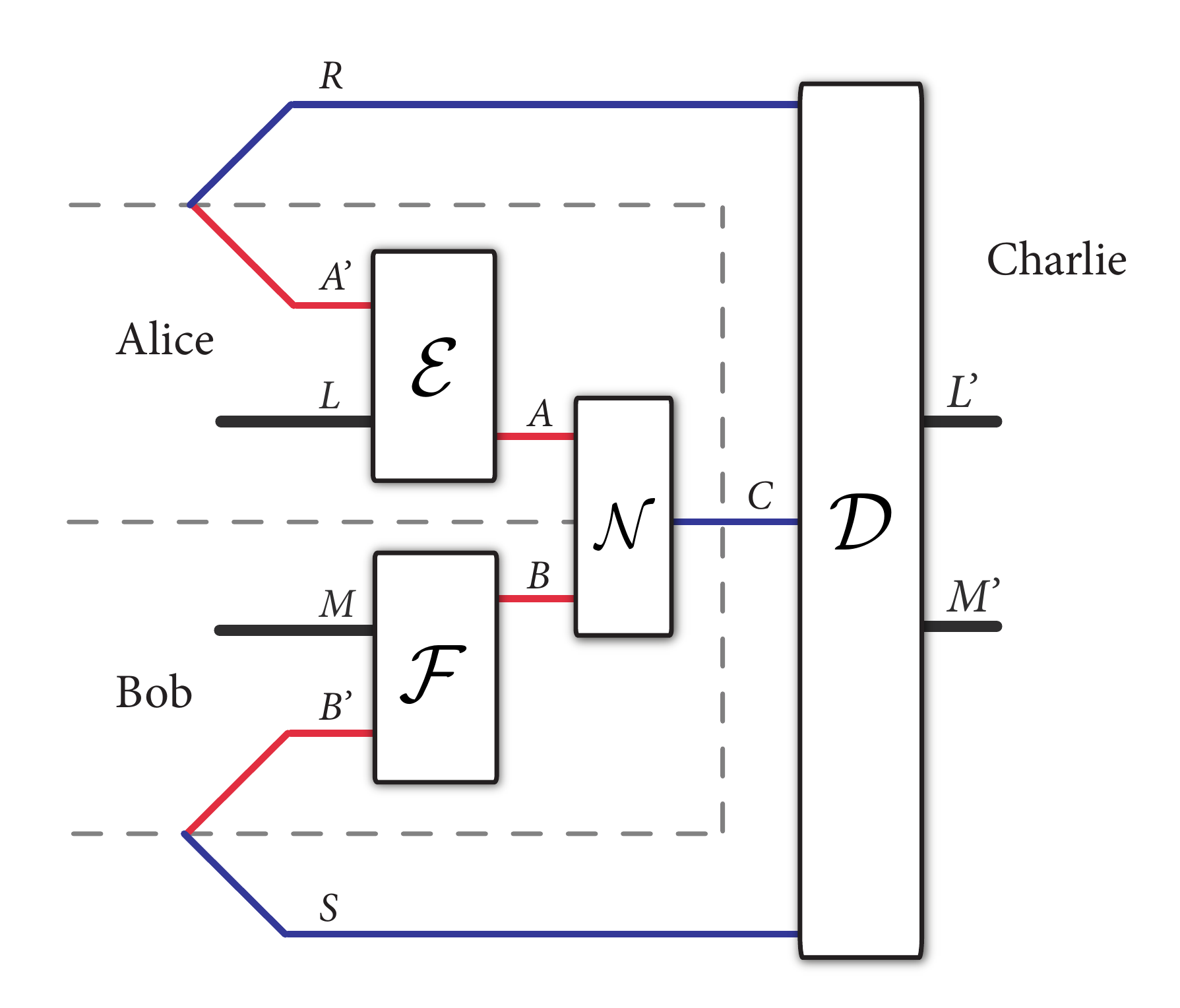}
\end{center}
\caption{The coding task for entanglement-assisted classical communication
over a quantum multiple-access channel. Alice and Bob each share entanglement
with the receiver Charlie, and Charlie employs a decoding channel
$\mathcal{D}$ to figure out which messages Alice and Bob transmitted.}%
\label{fig:EA-MAC-task}%
\end{figure}

\subsection{One-shot position-based coding scheme}

\label{sec:one-shot-position-QMAC}We now describe and analyze a position-based
coding scheme for entanglement-assisted communication over a quantum
multiple-access channel, in which the decoding POVM is a quantum simultaneous decoder.

\textbf{Encoding:} Before communication begins, suppose that Alice and Charlie
share $L$ copies of the same bipartite state: $\theta_{RA}^{\otimes L}%
\equiv\theta_{R_{1}A_{1}}\otimes\cdots\otimes\theta_{R_{L}A_{L}}$. Similarly,
suppose that Bob and Charlie share $M$ copies of the same bipartite state:
$\gamma_{SB}^{\otimes M}\equiv\gamma_{S_{1}B_{1}}\otimes\cdots\otimes
\gamma_{S_{M}B_{M}}$. To send message $(l,m)\in\mathcal{L}\times\mathcal{M}$,
Alice sends the $l$th $A$ system of $\theta_{RA}^{\otimes L}$ and Bob sends
the $m$th $B$ system of $\theta_{SB}^{\otimes M}$ over the quantum
multiple-access channel $\mathcal{N}_{AB\rightarrow C}$. So this leads to the
following state for Charlie:
\begin{equation}
\rho_{R^{L}S^{M}C}^{l,m}=\theta_{R}^{\otimes(l-1)}\otimes\gamma_{S}%
^{\otimes(m-1)}\otimes\mathcal{N}_{AB\rightarrow C}(\theta_{R_{l}A_{l}}%
\otimes\gamma_{S_{m}B_{m}})\otimes\theta_{R}^{\otimes(L-l)}\otimes\gamma
_{S}^{\otimes(M-m)}~. \label{eq:EA-MAC-state}%
\end{equation}

\textbf{Decoding:} To decode the message transmitted, Charlie performs a
measurement on the systems $R^{L}$, $S^{M}$, and the channel output $C$ to
determine the message pair $(l,m)$ that Alice and Bob transmitted. Consider
the following measurement operator:%
\begin{equation}
\Gamma_{R^{L}S^{M}C}^{l,m}\equiv T_{R_{l}S_{m}C}~, \label{eq:EA-MAC-meas}%
\end{equation}
where identity operators are implicit for all of the $R$ and $S$ systems
besides $R_{l}$ and $S_{m}$ and $T_{RSC}$ is a measurement operator satisfying
$0\leq T_{RSC}\leq I_{RSC}$. Let us call the action of performing the
measurement $\{\Gamma_{R^{L}S^{M}C}^{l,m},I_{R^{L}S^{M}C}-\Gamma_{R^{L}S^{M}%
C}^{l,m}\}$ \textquotedblleft checking for the message pair $(l,m)$%
.\textquotedblright\ If Charlie checks for message pair $(l,m)$ when indeed
message pair $(l,m)$ is transmitted, then the probability of incorrectly
decoding is%
\begin{equation}
\operatorname{Tr}\{(I-\Gamma_{R^{L}S^{M}C}^{l,m})\rho_{R^{L}S^{M}C}%
^{l,m}\}=\operatorname{Tr}\{(I-T_{RSC})\mathcal{N}_{AB\rightarrow C}%
(\theta_{RA}\otimes\gamma_{SB})\}. \label{eq:EA-MAC-err-1}%
\end{equation}
The equality above follows by combining \eqref{eq:EA-MAC-state}\ and
\eqref{eq:EA-MAC-meas} and\ applying partial traces. There are three other
kinds of error probabilities to consider. If message pair $(l,m)$ was
transmitted and Charlie checks for message pair $(l^{\prime},m)$ for
$l^{\prime}\neq l$, then the probability of decoding as $(l^{\prime},m)$ is%
\begin{equation}
\operatorname{Tr}\{\Gamma_{R^{L}S^{M}C}^{l^{\prime},m}\rho_{R^{L}S^{M}C}%
^{l,m}\}=\operatorname{Tr}\{T_{RSC}\mathcal{N}_{AB\rightarrow C}(\theta
_{R}\otimes\theta_{A}\otimes\gamma_{SB})\}. \label{eq:EA-MAC-err-2}%
\end{equation}
If message pair $(l,m)$ was transmitted and Charlie checks for message pair
$(l,m^{\prime})$ for $m^{\prime}\neq m$, then the probability of decoding as
$(l,m^{\prime})$ is%
\begin{equation}
\operatorname{Tr}\{\Gamma_{R^{L}S^{M}C}^{l,m^{\prime}}\rho_{R^{L}S^{M}C}%
^{l,m}\}=\operatorname{Tr}\{T_{RSC}\mathcal{N}_{AB\rightarrow C}(\theta
_{RA}\otimes\gamma_{S}\otimes\gamma_{B})\}. \label{eq:EA-MAC-err-3}%
\end{equation}
If message pair $(l,m)$ was transmitted and Charlie checks for message pair
$(l^{\prime},m^{\prime})$ for $l^{\prime}\neq l$ and $m^{\prime}\neq m$, then
the probability of decoding as $(l^{\prime},m^{\prime})$ is%
\begin{equation}
\operatorname{Tr}\{\Gamma_{R^{L}S^{M}C}^{l^{\prime},m^{\prime}}\rho
_{R^{L}S^{M}C}^{l,m}\}=\operatorname{Tr}\{T_{RSC}\mathcal{N}_{AB\rightarrow
C}(\theta_{R}\otimes\theta_{A}\otimes\gamma_{S}\otimes\gamma_{B})\}.
\label{eq:EA-MAC-err-4}%
\end{equation}
The above observations are helpful in the forthcoming error analysis.

We now take Charlie's position-based quantum simultaneous decoder to be the
following square-root measurement:%
\begin{equation}
\Lambda_{R^{L}S^{M}C}^{l,m}\equiv\left(  \sum_{l^{\prime}=1}^{L}%
\sum_{m^{\prime}=1}^{M}\Gamma_{R^{L}S^{M}C}^{l^{\prime},m^{\prime}}\right)
^{-1/2}\Gamma_{R^{L}S^{M}C}^{l,m}\left(  \sum_{l^{\prime}=1}^{L}%
\sum_{m^{\prime}=1}^{M}\Gamma_{R^{L}S^{M}C}^{l^{\prime},m^{\prime}}\right)
^{-1/2}. \label{eq:square-root-meas}%
\end{equation}

\textbf{Error analysis:} Due to the code construction, the error probability
under the position-based coding scheme is the same for each message pair
$(l,m)$:%
\begin{equation}
p_{e}(l,m)=\operatorname{Tr}\{(I-\Lambda_{R^{L}S^{M}C}^{l,m})\rho_{R^{L}%
S^{M}C}^{l,m}\}~.
\end{equation}
Applying Lemma~\ref{lemma:Hayashi-Nagaoka} and
\eqref{eq:EA-MAC-err-1}--\eqref{eq:EA-MAC-err-4}, we arrive at the following
upper bound on the decoding error probability:%
\begin{align}
p_{e}(l,m)  &  \leq c_{\operatorname{I}}\operatorname{Tr}\{(I-\Gamma
_{R^{L}S^{M}C}^{l,m})\rho_{R^{L}S^{M}C}^{l,m}\}+c_{\operatorname{II}}%
\sum_{(l^{\prime},m^{\prime})\neq(l,m)}\operatorname{Tr}\{\Gamma_{R^{L}S^{M}%
C}^{l^{\prime},m^{\prime}}\rho_{R^{L}S^{M}C}^{l,m}\}\\
&  =c_{\operatorname{I}}\operatorname{Tr}\{(I-\Gamma_{R^{L}S^{M}C}^{l,m}%
)\rho_{R^{L}S^{M}C}^{l,m}\}+c_{\operatorname{II}}\sum_{l^{\prime}\neq
l}\operatorname{Tr}\{\Gamma_{R^{L}S^{M}C}^{l^{\prime},m}\rho_{R^{L}S^{M}%
C}^{l,m}\}\nonumber\\
&  \qquad+c_{\operatorname{II}}\sum_{m^{\prime}\neq m}\operatorname{Tr}%
\{\Gamma_{R^{L}S^{M}C}^{l,m^{\prime}}\rho_{R^{L}S^{M}C}^{l,m}%
\}+c_{\operatorname{II}}\sum_{l^{\prime}\neq l,m^{\prime}\neq m}%
\operatorname{Tr}\{\Gamma_{R^{L}S^{M}C}^{l^{\prime},m^{\prime}}\rho
_{R^{L}S^{M}C}^{l,m}\}\\
&  =c_{\operatorname{I}}\operatorname{Tr}\{(I-T_{RSC})\mathcal{N}%
_{AB\rightarrow C}(\theta_{RA}\otimes\gamma_{SB})\}+c_{\operatorname{II}%
}\big[(L-1)\operatorname{Tr}\{T_{RSC}\mathcal{N}_{AB\rightarrow C}(\theta
_{R}\otimes\theta_{A}\otimes\gamma_{SB})\}\nonumber\\
&  \qquad+(M-1)\operatorname{Tr}\{T_{RSC}\mathcal{N}_{AB\rightarrow C}%
(\theta_{RA}\otimes\gamma_{S}\otimes\gamma_{B})\}\nonumber\\
&  \qquad+(L-1)(M-1)\operatorname{Tr}\{T_{RSC}\mathcal{N}_{AB\rightarrow
C}(\theta_{R}\otimes\theta_{A}\otimes\gamma_{S}\otimes\gamma_{B})\}\big]~.
\end{align}
The error probability associated with $c_{\operatorname{I}}$ is the
probability of incorrectly decoding when Charlie checks for message pair
$(l,m)$. The error probabilities associated with $c_{\operatorname{II}}$ are
the probabilities of decoding as some other message pair when message pair
$(l,m)$ is transmitted. There are $L-1$ possibilities for Charlie to
erroneously decode Alice's message and correctly decode Bob's message, $M-1$
possibilities to erroneously decode Bob's message and correctly decode Alice's
message, and $(M-1)(L-1)$ possibilities of incorrectly decoding both Alice and
Bob's messages.

\bigskip\textbf{One-shot bound for quantum simultaneous decoding.} Thus, our
bound on the decoding error probability for a position-based
entanglement-assisted coding scheme is as follows:%
\begin{multline}
p_{e}(l,m)\leq c_{\operatorname{I}}\operatorname{Tr}\{(I-T_{RSC}%
)\mathcal{N}_{AB\rightarrow C}(\theta_{RA}\otimes\gamma_{SB}%
)\}+c_{\operatorname{II}}\big[(L-1)\operatorname{Tr}\{T_{RSC}\mathcal{N}%
_{AB\rightarrow C}(\theta_{R}\otimes\theta_{A}\otimes\gamma_{SB})\}\\
+(M-1)\operatorname{Tr}\{T_{RSC}\mathcal{N}_{AB\rightarrow C}(\theta
_{RA}\otimes\gamma_{S}\otimes\gamma_{B})\}\\
+(L-1)(M-1)\operatorname{Tr}\{T_{RSC}\mathcal{N}_{AB\rightarrow C}(\theta
_{R}\otimes\theta_{A}\otimes\gamma_{S}\otimes\gamma_{B})\}\big]~.
\label{eq:EA-MAC-one-shot}%
\end{multline}
Interestingly, this bound is the same for all message pairs, and thus holds
for maximal or average error probability. We also remark that the above
inequality forges a transparent link between communication over
multiple-access channels and multiple quantum hypothesis testing, a point that
we will return to in Section~\ref{sec:disucssion}.

\bigskip\textbf{Generalization to multiple senders.} The above bound can be
extended as follows for an entanglement-assisted quantum multiple-access
channel $\mathcal{N}_{A_{1}\cdots A_{K}\rightarrow C}$\ with $K$ senders and a
single receiver:%
\begin{multline}
p_{e}(m_{1},\ldots,m_{K})\leq c_{\operatorname{I}}\operatorname{Tr}\left\{
(I-T_{R_{1}\cdots R_{K}C})\mathcal{N}_{A_{1}\cdots A_{K}\rightarrow C}\left(
\bigotimes\limits_{k=1}^{K}\theta_{R_{k}A_{k}}\right)  \right\}
\label{eq:avg-err-multiple-senders}\\
+c_{\operatorname{II}}\sum_{\mathcal{J}\subseteq\left[  K\right]  }\left[
\prod\limits_{j\in\mathcal{J}}\left(  M_{j}-1\right)  \right]
\operatorname{Tr}\left\{  T_{R_{1}\cdots R_{K}C}\mathcal{N}_{A_{1}\cdots
A_{K}\rightarrow C}\left(  \bigotimes\limits_{j\in\mathcal{J}}\theta_{R_{j}%
}\otimes\theta_{A_{j}}\otimes\bigotimes\limits_{l\in\mathcal{J}^{c}}%
\theta_{R_{l}A_{l}}\right)  \right\}  .
\end{multline}
In the above, $m_{1},\ldots,m_{K}$ are the messages for senders $1$ through
$K$, respectively, chosen from respective message sets of size $M_{1}%
,\ldots,M_{K}$. The states $\theta_{R_{1}A_{1}}$, \ldots, $\theta_{R_{K}A_{K}%
}$ are entangled states shared between the receiver and senders $1$ through
$K$, with the receiver possessing all of the $R$ systems. Finally,
$T_{R_{1}\cdots R_{K}C}$ is a test operator satisfying $0\leq T_{R_{1}\cdots
R_{K}C}\leq I_{R_{1}\cdots R_{K}C}$ and $\mathcal{J}$ is a non-empty subset of
$\left[  K\right]  \equiv\left\{  1,\ldots,K\right\}  $. The above bound is
derived by using position-based coding as described above and a square-root
measurement that generalizes \eqref{eq:square-root-meas}. We omit the
straightforward proof.

\subsection{Unassisted classical communication over multiple-access channels}

\label{sec:UA-MAC} The position-based coding technique is not only a powerful
tool for entanglement-assisted classical communication protocols, but also for
those that do not employ entanglement assistance or any other kind of
assistance. This was shown explicitly for the single-sender, single-receiver
case in \cite{wilde2017position}. We now demonstrate this point further by
considering unassisted classical communication over a classical-input
quantum-output multiple-access channel $\mathcal{N}_{XY\rightarrow C}$. We do
so by first considering classical communication assisted by shared randomness,
such that we can employ a position-based coding scheme, and then we
derandomize the protocol to obtain a codebook for unassisted communication.

The classical-classical-quantum channel that we consider can be written in
fully quantum form as%
\begin{equation}
\mathcal{N}_{X^{\prime}Y^{\prime}\rightarrow C}(\omega_{X^{\prime}Y^{\prime}%
})=\sum_{x,y}\langle x|_{X^{\prime}}\langle y|_{Y^{\prime}}\omega_{X^{\prime
}Y^{\prime}}|x\rangle_{X^{\prime}}|y\rangle_{Y^{\prime}}\ \rho_{C}^{x,y}~.
\end{equation}
Before communication begins, Alice and Charlie share randomness in the form of
the following classical--classical state:%
\begin{equation}
\rho_{XX^{\prime}}\equiv\sum_{x}p_{X}(x)|x\rangle\langle x|_{X}\otimes
|x\rangle\langle x|_{X^{\prime}}~.
\end{equation}
Similarly Bob and Charlie also share randomness represented by the following
classical--classical state:%
\begin{equation}
\sigma_{YY^{\prime}}\equiv\sum_{y}p_{Y}(y)|y\rangle\langle y|_{Y}%
\otimes|y\rangle\langle y|_{Y^{\prime}}~.
\end{equation}

We demonstrate the procedure of derandomization by proving the following theorem:

\begin{theorem}
\label{thm:assisted-to-unassisted}There exists an unassisted, simultaneous
decoding protocol for classical communication over a classical-input
quantum-output quantum multiple-access channel with the following upper bound
on its average error probability, holding for all $T_{XYC}$ such that $0\leq
T_{XYC}\leq I_{XYC}$:%
\begin{multline}
\frac{1}{LM}\sum_{l,m}p_{e}(l,m)\leq c_{\operatorname{I}}\operatorname{Tr}%
\{(I-T_{XYC})\mathcal{N}_{X^{\prime}Y^{\prime}\rightarrow C}(\rho_{X^{\prime
}X}\otimes\sigma_{YY^{\prime}})\}\\
+c_{\operatorname{II}}\big[\left(  L-1\right)  \operatorname{Tr}%
\{T_{XYC}\mathcal{N}_{X^{\prime}Y^{\prime}\rightarrow C}(\rho_{X}\otimes
\rho_{X^{\prime}}\otimes\sigma_{YY^{\prime}})\}\\
+\left(  M-1\right)  \operatorname{Tr}\{T_{XYC}\mathcal{N}_{X^{\prime
}Y^{\prime}\rightarrow C}(\rho_{XX^{\prime}}\otimes\sigma_{Y}\otimes
\sigma_{Y^{\prime}})\}\\
+\left(  L-1\right)  \left(  M-1\right)  \operatorname{Tr}\{T_{XYC}%
\mathcal{N}_{X^{\prime}Y^{\prime}\rightarrow C}(\rho_{X}\otimes\rho
_{X^{\prime}}\otimes\sigma_{Y}\otimes\sigma_{Y^{\prime}})\}\big]~,
\end{multline}
where $L$ is the number of messages for the first sender and $M$ is the number
of messages for the second sender. A generalization of this statement holds
for multiple senders, with an upper bound on the average error probability
given by the right-hand side of \eqref{eq:avg-err-multiple-senders}, but with
all of the $R$ and $A$ systems being classical.
\end{theorem}

\begin{proof}
The position-based coding scheme operates exactly as specified in
Section~\ref{sec:one-shot-position-QMAC}, with the states $\theta_{RA}$ and
$\gamma_{SB}$ replaced by $\rho_{XX^{\prime}}$ and $\sigma_{YY^{\prime}}$,
respectively, the channel $\mathcal{N}_{AB\rightarrow C}$ replaced by
$\mathcal{N}_{X^{\prime}Y^{\prime}\rightarrow C}$, and the test operator
$T_{RSC}$ replaced by $T_{XYC}$. The same error analysis then leads to the
following bound on the error probability when decoding the message pair
$(l,m)$:%
\begin{multline}
p_{e}(l,m)\leq c_{\operatorname{I}}\operatorname{Tr}\{(I-T_{XYC}%
)\mathcal{N}_{X^{\prime}Y^{\prime}\rightarrow C}(\rho_{XX^{\prime}}%
\otimes\sigma_{YY^{\prime}})\}\\
+c_{\operatorname{II}}\big[\left(  L-1\right)  \operatorname{Tr}%
\{T_{XYC}\mathcal{N}_{X^{\prime}Y^{\prime}\rightarrow C}(\rho_{X}\otimes
\rho_{X^{\prime}}\otimes\sigma_{YY^{\prime}})\}\\
+\left(  M-1\right)  \operatorname{Tr}\{T_{XYC}\mathcal{N}_{X^{\prime
}Y^{\prime}\rightarrow C}(\rho_{XX^{\prime}}\otimes\sigma_{Y}\otimes
\sigma_{Y^{\prime}})\}\\
+\left(  L-1\right)  \left(  M-1\right)  \operatorname{Tr}\{T_{XYC}%
\mathcal{N}_{X^{\prime}Y^{\prime}\rightarrow C}(\rho_{X}\otimes\rho
_{X^{\prime}}\otimes\sigma_{Y}\otimes\sigma_{Y^{\prime}})\}\big]~.
\label{eq:MAC-one-shot}%
\end{multline}

\textbf{Derandomization:} Extending the development in
\cite{wilde2017position}, first notice that we can rewrite the four trace
terms in \eqref{eq:MAC-one-shot} as follows:%
\begin{align}
\operatorname{Tr}\{T_{XYC}\mathcal{N}_{X^{\prime}Y^{\prime}\rightarrow C}%
(\rho_{XX^{\prime}}\otimes\sigma_{YY^{\prime}})\}  &  =\operatorname{Tr}%
\{T_{XYC}\sum_{x,y}p_{X}(x)p_{Y}(y)|xy\rangle\langle xy|\otimes\rho_{C}%
^{x,y}\}\\
&  =\sum_{x,y}p_{X}(x)p_{Y}(y)\operatorname{Tr}\{Q_{C}^{x,y}\rho_{C}%
^{x,y}\}~,\\
\operatorname{Tr}\{T_{XYC}\mathcal{N}_{X^{\prime}Y^{\prime}\rightarrow C}%
(\rho_{X}\otimes\rho_{X^{\prime}}\otimes\sigma_{YY^{\prime}})\}  &
=\sum_{x,y}p_{X}(x)p_{Y}(y)\operatorname{Tr}\{Q_{C}^{x,y}\bar{\rho}_{C}%
^{y}\}~,\\
\operatorname{Tr}\{T_{XYC}\mathcal{N}_{X^{\prime}Y^{\prime}\rightarrow C}%
(\rho_{XX^{\prime}}\otimes\sigma_{Y}\otimes\sigma_{Y^{\prime}})\}  &
=\sum_{x,y}p_{X}(x)p_{Y}(y)\operatorname{Tr}\{Q_{C}^{x,y}\bar{\rho}_{C}%
^{x}\}~,\\
\operatorname{Tr}\{T_{XYC}\mathcal{N}_{X^{\prime}Y^{\prime}\rightarrow C}%
(\rho_{X}\otimes\rho_{X^{\prime}}\otimes\sigma_{Y}\otimes\sigma_{Y^{\prime}%
})\}  &  =\sum_{x,y}p_{X}(x)p_{Y}(y)\operatorname{Tr}\{Q_{C}^{x,y}\bar{\rho
}_{C}\}~,
\end{align}
where we define the following averaged output states:%
\begin{align}
\bar{\rho}_{C}  &  \equiv\sum_{x,y}p_{X}(x)p_{Y}(y)\rho_{C}^{x,y},\\
\bar{\rho}_{C}^{y}  &  \equiv\sum_{x}p_{X}(x)\rho_{C}^{x,y},\\
\bar{\rho}_{C}^{x}  &  \equiv\sum_{y}p_{Y}(y)\rho_{C}^{x,y},
\end{align}
and the measurement operator%
\begin{equation}
Q_{C}^{x,y}\equiv\langle x,y|_{XY}T_{XYC}|x,y\rangle_{XY}~.
\end{equation}
Thus, in the case that the code is randomness-assisted, it suffices to take
the test operator $T_{XYC}$ to have the following form:%
\begin{equation}
T_{XYC}=\sum_{x,y}|x\rangle\langle x|_{X}\otimes|y\rangle\langle y|_{Y}\otimes
Q_{C}^{x,y},
\end{equation}
because, as we will show, the upper bound on the average error probability
does not change when doing so. Then we can rewrite the decoding POVM in
\eqref{eq:square-root-meas}\ as follows:%
\begin{align}
\Gamma_{X^{L}Y^{M}C}^{l,m}  &  \equiv T_{X_{l}Y_{m}C}\\
&  =\sum_{x^{L},y^{M}}|x^{L}\rangle\langle x^{L}|_{X^{L}}\otimes|y^{M}%
\rangle\langle y^{M}|_{Y^{M}}\otimes Q_{C}^{x_{l},y_{m}}~,
\end{align}
where we use the resolution of the identity $I_{X}=\sum_{x}|x\rangle\langle
x|_{X}$ to expand the implicit identity operators and we employ the notation
$x^{L}\equiv x_{1}\cdots x_{L}$ and $y^{M}\equiv y_{1}\cdots y_{M}$. Then this
implies that
\begin{align}
\left(  \sum_{l^{\prime}=1}^{L}\sum_{m^{\prime}=1}^{M}\Gamma_{X^{L}Y^{M}%
C}^{l^{\prime},m^{\prime}}\right)  ^{-1/2}  &  =\left(  \sum_{l^{\prime}%
=1}^{L}\sum_{m^{\prime}=1}^{M}\sum_{x^{L},y^{M}}|x^{L}\rangle\langle
x^{L}|_{X^{L}}\otimes|y^{M}\rangle\langle y^{M}|_{Y^{M}}\otimes Q_{C}%
^{x_{l^{\prime}},y_{m^{\prime}}}\right)  ^{-1/2}\\
&  =\left(  \sum_{x^{L},y^{M}}|x^{L}\rangle\langle x^{L}|_{X^{L}}\otimes
|y^{M}\rangle\langle y^{M}|_{Y^{M}}\otimes\sum_{l^{\prime}=1}^{L}%
\sum_{m^{\prime}=1}^{M}Q_{C}^{x_{l^{\prime}},y_{m^{\prime}}}\right)  ^{-1/2}\\
&  =\sum_{x^{L},y^{M}}|x^{L}\rangle\langle x^{L}|_{X^{L}}\otimes|y^{M}%
\rangle\langle y^{M}|_{Y^{M}}\otimes\left(  \sum_{l^{\prime}=1}^{L}%
\sum_{m^{\prime}=1}^{M}Q_{C}^{x_{l^{\prime}},y_{m^{\prime}}}\right)  ^{-1/2}~.
\end{align}
The last step follows from the fact that $\{|x\rangle\}_{x}$ and
$\{|y\rangle\}_{y}$ form orthonormal bases. Therefore, the decoding POVM for
the randomness-assisted protocol can be decomposed as
\begin{align}
\Lambda_{X^{L}Y^{M}C}^{l,m}  &  \equiv\left(  \sum_{l^{\prime}=1}^{L}%
\sum_{m^{\prime}=1}^{M}\Gamma_{X^{L}Y^{M}C}^{l^{\prime},m^{\prime}}\right)
^{-1/2}\Gamma_{X^{L}Y^{M}C}^{l,m}\left(  \sum_{l^{\prime}=1}^{L}%
\sum_{m^{\prime}=1}^{M}\Gamma_{X^{L}Y^{M}C}^{l^{\prime},m^{\prime}}\right)
^{-1/2}~,\\
&  =\sum_{x^{L},y^{M}}|x^{L}\rangle\langle x^{L}|_{X^{L}}\otimes|y^{M}%
\rangle\langle y^{M}|_{Y^{M}}\otimes\Omega_{C}^{x_{l},y_{m}}~,
\label{eq:unassisted-MAC-POVM-expansion}%
\end{align}
where%
\begin{equation}
\Omega_{C}^{x_{l},y_{m}}\equiv\left(  \sum_{l^{\prime}=1}^{L}\sum_{m^{\prime
}=1}^{M}Q_{C}^{x_{l^{\prime}},y_{m^{\prime}}}\right)  ^{-1/2}Q_{C}%
^{x_{l},y_{m}}\left(  \sum_{l^{\prime}=1}^{L}\sum_{m^{\prime}=1}^{M}%
Q_{C}^{x_{l^{\prime}},y_{m^{\prime}}}\right)  ^{-1/2}~.
\end{equation}
By definition, the output state of Charlie in \eqref{eq:EA-MAC-state}, for our
case of interest, can be written as%
\begin{equation}
\rho_{X^{L}Y^{M}C}^{l,m}=\sum_{x^{L},y^{M}}p_{X^{L}}(x^{L})p_{Y^{M}}%
(y^{M})|x^{L}\rangle\langle x^{L}|_{X^{L}}\otimes|y^{M}\rangle\langle
y^{M}|_{Y^{M}}\otimes\rho_{C}^{x_{l},y_{m}}~.
\label{eq:unassisted-MAC-expanded-state}%
\end{equation}
By combining \eqref{eq:unassisted-MAC-POVM-expansion} and
\eqref{eq:unassisted-MAC-expanded-state}, we find that the average error
probability for the code can be rewritten as%
\begin{align}
&  \frac{1}{LM}\sum_{l=1}^{L}\sum_{m=1}^{M}\operatorname{Tr}\{(I_{X^{L}Y^{M}%
C}-\Lambda_{X^{L}Y^{M}C}^{l,m})\rho_{X^{L}Y^{M}C}^{l,m}\}\nonumber\\
&  =\frac{1}{LM}\sum_{l=1}^{L}\sum_{m=1}^{M}\sum_{x^{L},y^{M}}p_{X^{L}}%
(x^{L})p_{Y^{M}}(y^{M})\operatorname{Tr}\{(I_{C}-\Omega_{C}^{x_{l},y_{m}}%
)\rho_{C}^{x_{l},y_{m}}\}\label{eq:derand-1}\\
&  =\sum_{x^{L},y^{M}}p_{X^{L}}(x^{L})p_{Y^{M}}(y^{M})\left[  \frac{1}{LM}%
\sum_{l=1}^{L}\sum_{m=1}^{M}\operatorname{Tr}\{(I_{C}-\Omega_{C}^{x_{l},y_{m}%
})\rho_{C}^{x_{l},y_{m}}\}\right]  . \label{eq:derand-2}%
\end{align}
Suppose now that there exists a randomness-assisted position-based code that
has an average error probability $\leq\varepsilon$. By the above equalities
and since the average can never exceed the maximum, we know there must exist a
particular choice of $x^{L},y^{L}$ such that
\begin{equation}
\frac{1}{LM}\sum_{l=1}^{L}\sum_{m=1}^{M}\operatorname{Tr}\{(I-\Omega
_{C}^{x_{l},y_{m}})\rho_{C}^{x_{l},y_{m}}\}\leq\varepsilon~.
\end{equation}
Thus for an unassisted multiple-access classical communication protocol, if we
choose $\{x_{l}\}_{l=1}^{L}$ as Alice's codebook and $\{y_{m}\}_{m=1}^{M}$ as
Bob's codebook, and the POVM $\{\Omega_{c}^{x_{l},y_{m}}\}$ as Charlie's
decoder, an upper bound on the average probability error is given by
\begin{equation}
\frac{1}{LM}\sum_{l=1}^{L}\sum_{m=1}^{M}p_{e}(l,m)=\frac{1}{LM}\sum_{l=1}%
^{L}\sum_{m=1}^{M}\operatorname{Tr}\{(I-\Omega_{C}^{x_{l},y_{m}})\rho
_{C}^{x_{l},y_{m}}\}\leq\varepsilon~.
\end{equation}
This proves the statement of the theorem after considering the upper bound in \eqref{eq:MAC-one-shot}.
\end{proof}

\subsection{Achievable rate region for i.i.d.~case}

\label{subsubsec:EA-MAC-iid}

We now demonstrate rates that are achievable when using a particular quantum
simultaneous decoder. Interestingly, we show how the same quantum simultaneous
decoder leads to two generally different bounds for the achievable rate
region. In the first one, the rates are bounded by terms which consist of the
difference of a R\'{e}nyi entropy of order two and a conditional quantum
entropy. In the second one, the rates are bounded by terms which consist of
the difference of a collision conditional entropy and a conditional quantum
entropy. Although these rates are suboptimal when compared to what is
achievable by using successive decoding
\cite{winter2001capacity,hsieh2008entanglement}, we nevertheless think that
the following coding theorems represent progress toward finding a quantum
simultaneous decoder.

Before we state the theorems, we require the following definition:\ A rate
pair $(R_{1},R_{2})$ is achievable for communication over a quantum multiple
access channel if there exists a $(2^{n\left[  R_{1}-\delta\right]
},2^{n\left[  R_{2}-\delta\right]  },\varepsilon)$ code for communication over
$\mathcal{N}_{A^{\prime}B^{\prime}\rightarrow C}^{\otimes n}$\ for all
$\varepsilon\in(0,1)$, $\delta>0$, and sufficiently large $n$.

\begin{theorem}
\label{thm:achieve-simul-renyi-2}An achievable rate region for
entanglement-assisted classical communication over a two-sender quantum
multiple-access channel $\mathcal{N}_{AB\rightarrow C}$, by employing a
quantum simultaneous decoder, is as follows:%
\begin{align}
R_{1}  &  \leq\tilde{I}(S;CR)_{\omega}~,\\
R_{2}  &  \leq\tilde{I}(R;CS)_{\omega}~,\\
R_{1}+R_{2}  &  \leq\tilde{I}(RC;S)_{\omega}~,
\end{align}
where $\omega_{RSC}\equiv\mathcal{N}_{AB\rightarrow C}(\theta_{RA}%
\otimes\gamma_{SB})$ and $\theta_{RA}$ and $\gamma_{SB}$ are quantum states.
Here we define the following mutual-information-like quantities:
\begin{align}
\tilde{I}(R;CS)_{\omega}  &  \equiv H_{2}(SC)_{\omega}-H(SC|R)_{\omega
}~,\label{eq:renyi-two-quant-1}\\
\tilde{I}(S;CR)_{\omega}  &  \equiv H_{2}(RC)_{\omega}-H(RC|S)_{\omega}~,\\
\tilde{I}(RS;C)_{\omega}  &  \equiv H_{2}(C)_{\omega}-H(C|RS)_{\omega}~,
\label{eq:renyi-two-quant-3}%
\end{align}
where $H_{2}(A)_{\rho}\equiv-\log_{2}\operatorname{Tr}\{\rho_{A}^{2}\}$ is the
R\'{e}nyi entropy of order two.
\end{theorem}

\begin{proof}
In our setting, Alice and Bob use an i.i.d.~channel $\mathcal{N}_{A^{\prime
}B^{\prime}\rightarrow C}^{\otimes n}$. In order to bound the error
probability, we replace each state in \eqref{eq:EA-MAC-one-shot} by its
$n$-copy version. Defining $\omega_{RSC}\equiv\mathcal{N}_{AB\rightarrow
C}(\theta_{RA}\otimes\gamma_{SB})$, we choose the test operator $T$ to be the
following `coated' typical projector:%
\begin{equation}
T_{R^{n}S^{n}C^{n}}\equiv(\Pi_{R^{n}}^{\omega,\delta}\otimes\Pi_{S^{n}%
}^{\omega,\delta})\Pi_{R^{n}S^{n}C^{n}}^{\omega,\delta}(\Pi_{R^{n}}%
^{\omega,\delta}\otimes\Pi_{S^{n}}^{\omega,\delta})~,
\label{eq:renyi-2-simul-decode}%
\end{equation}
where $\Pi_{R^{n}}^{\omega,\delta}$, $\Pi_{S^{n}}^{\omega,\delta}$, and
$\Pi_{R^{n}S^{n}C^{n}}^{\omega,\delta}$ are the typical projectors
corresponding to the respective states $\omega_{R}$, $\omega_{S}$, and
$\omega_{RSC}$. Applying \eqref{eq:EA-MAC-one-shot}, we find the following
upper bound on the error probability when decoding the message pair $\left(
l,m\right)  $:%
\begin{multline}
p_{e}(l,m)\leq c_{\operatorname{I}}\operatorname{Tr}\{(I-T_{R^{n}S^{n}C^{n}%
})\omega_{RSC}^{\otimes n}\}+c_{\operatorname{II}}\big[L\operatorname{Tr}%
\{T_{R^{n}S^{n}C^{n}}[\mathcal{N}_{AB\rightarrow C}(\theta_{R}\otimes
\theta_{A}\otimes\gamma_{SB})]^{\otimes n})\}\\
+M\operatorname{Tr}\{T_{R^{n}S^{n}C^{n}}[\mathcal{N}_{AB\rightarrow C}%
(\theta_{RA}\otimes\gamma_{S}\otimes\gamma_{B})]^{\otimes n}\}\\
+LM\operatorname{Tr}\{T_{R^{n}S^{n}C^{n}}[\mathcal{N}_{AB\rightarrow C}%
(\theta_{R}\otimes\theta_{A}\otimes\gamma_{S}\otimes\gamma_{B})]^{\otimes
n}\}\big]~. \label{eq:err_bound}%
\end{multline}
We evaluate each term sequentially. To give an upper bound on the first term,
consider the following chain of inequalities, which holds for sufficiently
large $n$:%
\begin{align}
&  \operatorname{Tr}\{(\Pi_{R^{n}}^{\omega,\delta}\otimes\Pi_{S^{n}}%
^{\omega,\delta})\Pi_{R^{n}S^{n}C^{n}}^{\omega,\delta}(\Pi_{R^{n}}%
^{\omega,\delta}\otimes\Pi_{S^{n}}^{\omega,\delta})\omega_{RSC}^{\otimes
n}\}\nonumber\\
&  \geq\operatorname{Tr}\{\Pi_{R^{n}S^{n}C^{n}}^{\omega,\delta}\omega
_{RSC}^{\otimes n}\}-\left\Vert (\Pi_{R^{n}}^{\omega,\delta}\otimes\Pi_{S^{n}%
}^{\omega,\delta})\omega_{RSC}^{\otimes n}(\Pi_{R^{n}}^{\omega,\delta}%
\otimes\Pi_{S^{n}}^{\omega,\delta})-\omega_{RSC}^{\otimes n}\right\Vert _{1}\\
&  \geq1-\varepsilon-2\sqrt{2\varepsilon}~.
\end{align}
The first inequality follows from Lemma~\ref{lemma:close}. The second
inequality follows from \eqref{ieq:typical-unit-prob}, \cite[Eq.~(81)]%
{hsieh2008entanglement}, and the application of Lemma~\ref{lemma:gentle}. We
then obtain an upper bound on the first term:%
\begin{equation}
\operatorname{Tr}\{(I-T_{R^{n}S^{n}C^{n}})\omega_{RSC}^{\otimes n}%
\}\leq\varepsilon+2\sqrt{2\varepsilon}~.
\end{equation}

Now we consider the second term in \eqref{eq:err_bound}:%
\begin{align}
&  L\operatorname{Tr}\{(\Pi_{R^{n}}^{\omega,\delta}\otimes\Pi_{S^{n}}%
^{\omega,\delta})\Pi_{R^{n}S^{n}C^{n}}^{\omega,\delta}(\Pi_{R^{n}}%
^{\omega,\delta}\otimes\Pi_{S^{n}}^{\omega,\delta})[\mathcal{N}_{AB\rightarrow
C}(\theta_{R}\otimes\theta_{A}\otimes\gamma_{SB})]^{\otimes n})\}\nonumber\\
&  \leq L\ 2^{n[H(RSC)_{\omega}+\delta]}\operatorname{Tr}\{(\Pi_{R^{n}%
}^{\omega,\delta}\otimes\Pi_{S^{n}}^{\omega,\delta})\omega_{RSC}^{\otimes
n}(\Pi_{R^{n}}^{\omega,\delta}\otimes\Pi_{S^{n}}^{\omega,\delta})\theta
_{R}^{\otimes n}\otimes\lbrack\mathcal{N}_{AB\rightarrow C}(\theta_{A}%
\otimes\gamma_{SB})]^{\otimes n})\}\\
&  =L\ 2^{n[H(RSC)_{\omega}+\delta]}\operatorname{Tr}\{\Pi_{S^{n}}%
^{\omega,\delta}\omega_{RSC}^{\otimes n}\Pi_{S^{n}}^{\omega,\delta}\left[
(\Pi_{R^{n}}^{\omega,\delta}\theta_{R}^{\otimes n}\Pi_{R^{n}}^{\omega,\delta
})\otimes\lbrack\mathcal{N}_{AB\rightarrow C}(\theta_{A}\otimes\gamma
_{SB})]^{\otimes n}\right]  \}\\
&  \leq L\ 2^{n[H(RSC)_{\omega}+\delta]}2^{-n[H(R)_{\omega}-\delta
]}\operatorname{Tr}\{\Pi_{S^{n}}^{\omega,\delta}\omega_{RSC}^{\otimes n}%
\Pi_{S^{n}}^{\omega,\delta}\left[  I_{R^{n}}\otimes\lbrack\mathcal{N}%
_{AB\rightarrow C}(\theta_{A}\otimes\gamma_{SB})]^{\otimes n}\right]  \}\\
&  \leq L\ 2^{n[H(SC|R)_{\omega}+2\delta]}\operatorname{Tr}\{\Pi_{S^{n}%
}^{\omega,\delta}\omega_{SC}^{\otimes n}\Pi_{S^{n}}^{\omega,\delta}\omega
_{SC}^{\otimes n}\}\label{eq:cond-coll-pickup}\\
&  \leq L\ 2^{n[H(SC|R)_{\omega}+2\delta]}\operatorname{Tr}\{(\omega
_{SC}^{\otimes n})^{2}\}\\
&  =2^{nR_{1}}2^{n[H(SC|R)_{\omega}+2\delta]}2^{-nH_{2}(SC)_{\omega}}\\
&  =2^{-n[H_{2}(SC)_{\omega}-H(SC|R)_{\omega}-R_{1}-2\delta]}~.
\end{align}
The first inequality follows from the application of the projector trick
inequality from \eqref{ieq:project-trick} to the state $\omega_{RSC}^{\otimes
n}$. The first equality follows from cyclicity of trace. The second inequality
follows from the right-hand side of \eqref{ieq:typical-equal-partition}, the
fact that $\theta_{R}=\omega_{R}$, and the inequality $\Pi_{R^{n}}%
^{\omega,\delta}\leq I_{R^{n}}$. The third inequality follows from a partial
trace over the $R^{n}$ systems. The fourth inequality follows because%
\begin{equation}
\operatorname{Tr}\{\Pi_{S^{n}}^{\omega,\delta}\omega_{SC}^{\otimes n}%
\Pi_{S^{n}}^{\omega,\delta}\omega_{SC}^{\otimes n}\}\leq\operatorname{Tr}%
\{\omega_{SC}^{\otimes n}\Pi_{S^{n}}^{\omega,\delta}\omega_{SC}^{\otimes
n}\}=\operatorname{Tr}\{(\omega_{SC}^{\otimes n})^{2}\Pi_{S^{n}}%
^{\omega,\delta}\}\leq\operatorname{Tr}\{(\omega_{SC}^{\otimes n})^{2}\},
\end{equation}
which is a consequence of the facts that $\Pi_{S^{n}}^{\omega,\delta}\leq
I_{S^{n}}$ and $\omega_{SC}^{\otimes n}\Pi_{S^{n}}^{\omega,\delta}\omega
_{SC}^{\otimes n}$ and $(\omega_{SC}^{\otimes n})^{2}$ are positive
semi-definite. Finally, the second equality follows from the fact
$L=2^{nR_{1}}$ and the definition of R\'{e}nyi entropy of order two.

Following a similar analysis, we obtain the following upper bounds for the
other two terms:%
\begin{align}
M\operatorname{Tr}\{T_{R^{n}S^{n}C^{n}}[\mathcal{N}_{AB\rightarrow C}%
(\theta_{RA}\otimes\gamma_{S}\otimes\gamma_{B})]^{\otimes n}\}  &
\leq2^{-n[H_{2}(RC)_{\omega}-H(RC|S)_{\omega}-R_{2}-2\delta]}~,\\
LM\operatorname{Tr}\{T_{R^{n}S^{n}C^{n}}[\mathcal{N}_{AB\rightarrow C}%
(\theta_{R}\otimes\theta_{A}\otimes\gamma_{S}\otimes\gamma_{B})]^{\otimes
n}\}  &  \leq2^{-n[H_{2}(C)_{\omega}-H(C|RS)_{\omega}-(R_{1}+R_{2})-3\delta
]}~.
\end{align}

Taking the sum of the upper bounds for the above four terms, we find the
following upper bound on the error probability when decoding the message pair
$(l,m)$:%
\begin{multline}
p_{e}(l,m)\leq\varepsilon+2\sqrt{2\varepsilon}+2^{-n[H_{2}(SC)_{\omega
}-H(SC|R)_{\omega}-R_{1}-2\delta]}\\
+2^{-n[H_{2}(RC)_{\omega}-H(RC|S)_{\omega}-R_{2}-2\delta]}+2^{-n[H_{2}%
(C)_{\omega}-H(C|RS)_{\omega}-(R_{1}+R_{2})-3\delta]}~.
\end{multline}
Thus, if the rate pair $(R_{1},R_{2})$ satisfies the following inequalities
(related to those in the statement of the theorem)%
\begin{align}
R_{1}+3\delta &  \leq\tilde{I}(S;CR)_{\omega}~,\\
R_{2}+3\delta &  \leq\tilde{I}(R;CS)_{\omega}~,\\
R_{1}+R_{2}+4\delta &  \leq\tilde{I}(RC;S)_{\omega}~,
\end{align}
then the error probability can be made arbitrarily small with increasing $n$.
However, since $\delta$ can be taken arbitrarily small after the limit of
large $n$, we conclude that the rate region given in the statement of the
theorem is achievable.
\end{proof}

Our results can be easily extended to the multiple-sender scenario, which we
state below without explicitly writing down a proof (note that the proof is a
straightforward generalization of the above analysis for two senders).

\begin{theorem}
\label{thm:mult-send-gen}An achievable rate region for entanglement-assisted
classical communication over a $K$-sender multiple-access quantum channel
$\mathcal{N}_{A_{1}A_{2},...A_{K}\rightarrow C}$, by employing a quantum
simultaneous decoder, is given by the following:%
\begin{equation}
\sum_{j\in\mathcal{J}}R_{j}\leq\tilde{I}(S(\mathcal{J});CS(\mathcal{J}%
^{c}))_{\omega},\qquad\text{for every}~\mathcal{J}\subseteq\lbrack K]~,
\label{eq:info-bound-1-mac}%
\end{equation}
where $\omega_{S_{1}\cdots S_{K}C}\equiv\mathcal{N}_{A_{1}\cdots
A_{K}\rightarrow C}(\theta_{S_{1}A_{1}}\otimes\cdots\otimes\theta_{S_{K}A_{K}%
})$. Here we define mutual-information-like quantities
\begin{equation}
\tilde{I}(S(\mathcal{J});CS(\mathcal{J}^{c}))_{\omega}\equiv H_{2}%
(S(\mathcal{J}^{c})C)_{\omega}-H(S(\mathcal{J}^{c})C|S(\mathcal{J}))_{\omega
}~,
\end{equation}
where $H_{2}(B)_{\rho}=-\log_{2}\operatorname{Tr}\{\rho_{B}^{2}\}$ is the
R\'{e}nyi entropy of order two.
\end{theorem}

By combining Theorems~\ref{thm:assisted-to-unassisted}\ and
\ref{thm:mult-send-gen}, we obtain the following rate region that is
achievable for unassisted classical communication over a quantum
multiple-access channel when using a quantum simultaneous decoder:

\begin{theorem}
\label{thm:cc-mac-renyi-gen}The following rate region is achievable when using
a quantum simultaneous decoder for unassisted classical communication over the
$K$-sender, classical-input quantum-output multiple-access channel
$x_{1},\ldots,x_{K}\rightarrow\rho_{x_{1},\ldots,x_{K}}$:%
\begin{equation}
\sum_{j\in\mathcal{J}}R_{j}\leq\tilde{I}(X(\mathcal{J});CX(\mathcal{J}%
^{c}))_{\omega},\qquad\text{for every}~\mathcal{J}\subseteq\lbrack K]~,
\end{equation}
where%
\begin{equation}
\omega_{X_{1}\cdots X_{K}C}\equiv\sum_{x_{1},\ldots,x_{K}}p_{X_{1}}%
(x_{1})\cdots p_{X_{K}}(x_{K})|x_{1}\rangle\langle x_{1}|_{X_{1}}\otimes
\cdots\otimes|x_{K}\rangle\langle x_{K}|_{X_{K}}\otimes\rho_{x_{1}%
,\ldots,x_{K}}.
\end{equation}

\end{theorem}

The following alternative achievable rate region is generally different from
the one in Theorem~\ref{thm:mult-send-gen}:

\begin{theorem}
\label{thm:alt-simul-decode}An achievable rate region for
entanglement-assisted classical communication over the $K$-sender quantum
multiple-access channel $\mathcal{N}_{A_{1}A_{2},...A_{K}\rightarrow C}$, by
employing a quantum simultaneous decoder, is given by the following:%
\begin{equation}
\sum_{j\in\mathcal{J}}R_{j}\leq I^{\prime}(S(\mathcal{J});CS(\mathcal{J}%
^{c}))_{\omega},\qquad\text{for every}~\mathcal{J}\subseteq\lbrack K]~,
\label{eq:info-bound-2-mac}%
\end{equation}
where $\omega_{S_{1}\cdots S_{K}C}\equiv\mathcal{N}_{A_{1}\cdots
A_{K}\rightarrow C}(\theta_{S_{1}A_{1}}\otimes\cdots\otimes\theta_{S_{K}A_{K}%
})$. Here we define mutual-information-like quantities
\begin{equation}
I^{\prime}(S(\mathcal{J});CS(\mathcal{J}^{c}))_{\omega}\equiv H_{2}%
(C|S(\mathcal{J}^{c}))_{\omega}-H(C|S_{1}\cdots S_{K})_{\omega}~,
\end{equation}
where $H_{2}(A|B)_{\rho}=-\log_{2}\operatorname{Tr}\{\rho_{AB}\rho_{B}%
^{-1/2}\rho_{AB}\rho_{B}^{-1/2}\}$ is the collision conditional entropy
\cite{DFW15}.
\end{theorem}

\begin{proof}
We prove this theorem for the case of two senders, and then the extension to
three or more senders is straightforward. The analysis proceeds similarly as
in the proof of Theorem~\ref{thm:achieve-simul-renyi-2}, and we pick up at
\eqref{eq:cond-coll-pickup}. We find that%
\begin{align}
&  L\ 2^{n[H(SC|R)_{\omega}+2\delta]}\operatorname{Tr}\{\Pi_{S^{n}}%
^{\omega,\delta}\omega_{SC}^{\otimes n}\Pi_{S^{n}}^{\omega,\delta}\omega
_{SC}^{\otimes n}\}\nonumber\\
&  \leq L\ 2^{n[H(SC|R)_{\omega}+2\delta]}2^{-n\left[  H(S)_{\omega}%
-\delta\right]  }\operatorname{Tr}\{(\omega_{S}^{\otimes n})^{-1/2}\omega
_{SC}^{\otimes n}(\omega_{S}^{\otimes n})^{-1/2}\omega_{SC}^{\otimes n}\}\\
&  =2^{nR_{1}}2^{n[H(C|RS)_{\omega}+3\delta]}\left[  \operatorname{Tr}%
\{\omega_{S}^{-1/2}\omega_{SC}\omega_{S}^{-1/2}\omega_{SC}\}\right]
^{n}\label{eq:penul}\\
&  =2^{-n[H_{2}(C|S)_{\omega}-H(C|RS)_{\omega}-R_{1}-3\delta]}~.
\label{eq:equal-cond-coll-ent}%
\end{align}
The inequality follows because%
\begin{align}
&  \operatorname{Tr}\{\Pi_{S^{n}}^{\omega,\delta}\omega_{SC}^{\otimes n}%
\Pi_{S^{n}}^{\omega,\delta}\omega_{SC}^{\otimes n}\}\nonumber\\
&  \leq2^{-n\left[  H(S)_{\omega}-\delta\right]  /2}\operatorname{Tr}%
\{(\omega_{S}^{\otimes n})^{-1/2}\omega_{SC}^{\otimes n}\Pi_{S^{n}}%
^{\omega,\delta}\omega_{SC}^{\otimes n}\}\\
&  =2^{-n\left[  H(S)_{\omega}-\delta\right]  /2}\operatorname{Tr}%
\{\omega_{SC}^{\otimes n}(\omega_{S}^{\otimes n})^{-1/2}\omega_{SC}^{\otimes
n}\Pi_{S^{n}}^{\omega,\delta}\}\\
&  \leq2^{-n\left[  H(S)_{\omega}-\delta\right]  /2}2^{-n\left[  H(S)_{\omega
}-\delta\right]  /2}\operatorname{Tr}\{\omega_{SC}^{\otimes n}(\omega
_{S}^{\otimes n})^{-1/2}\omega_{SC}^{\otimes n}(\omega_{S}^{\otimes n}%
)^{-1/2}\},
\end{align}
where we apply \eqref{eq:sqrt-root-proj-trick} twice and the facts that
$\omega_{SC}^{\otimes n}\Pi_{S^{n}}^{\omega,\delta}\omega_{SC}^{\otimes n}%
\geq0$ and $\omega_{SC}^{\otimes n}(\omega_{S}^{\otimes n})^{-1/2}\omega
_{SC}^{\otimes n}\geq0$. The equality in \eqref{eq:penul} follows because
$H(S)_{\omega}=H(S|R)_{\omega}$ for the state $\omega$ and then because
$H(SC|R)_{\omega}-H(S|R)_{\omega}=H(C|RS)_{\omega}$. The equality in
\eqref{eq:equal-cond-coll-ent} follows from the definition of the collision
conditional entropy.

Following a similar analysis, we obtain the following upper bounds for the
other two terms:%
\begin{align}
M\operatorname{Tr}\{T_{R^{n}S^{n}C^{n}}[\mathcal{N}_{AB\rightarrow C}%
(\theta_{RA}\otimes\gamma_{S}\otimes\gamma_{B})]^{\otimes n}\}  &
\leq2^{-n[H_{2}(C|R)_{\omega}-H(C|RS)_{\omega}-R_{2}-3\delta]}~,\\
LM\operatorname{Tr}\{T_{R^{n}S^{n}C^{n}}[\mathcal{N}_{AB\rightarrow C}%
(\theta_{R}\otimes\theta_{A}\otimes\gamma_{S}\otimes\gamma_{B})]^{\otimes
n}\}  &  \leq2^{-n[H_{2}(C)_{\omega}-H(C|RS)_{\omega}-(R_{1}+R_{2})-3\delta
]}~.
\end{align}
The rest of the proof is then the same as in the proof of
Theorem~\ref{thm:achieve-simul-renyi-2}.
\end{proof}

By combining Theorems~\ref{thm:assisted-to-unassisted}\ and
\ref{thm:alt-simul-decode}, we obtain the following rate region that is
achievable for unassisted classical communication over a quantum
multiple-access channel when using a quantum simultaneous decoder:

\begin{theorem}
\label{thm:mult-send-gen-alt}The following rate region is achievable when
using a quantum simultaneous decoder for unassisted classical communication
over the $K$-sender, classical-input quantum-output multiple-access channel
$x_{1},\ldots,x_{K}\rightarrow\rho_{x_{1},\ldots,x_{K}}$:%
\begin{equation}
\sum_{j\in\mathcal{J}}R_{j}\leq I^{\prime}(X(\mathcal{J});CX(\mathcal{J}%
^{c}))_{\omega},\qquad\text{for every}~\mathcal{J}\subseteq\lbrack K]~,
\end{equation}
where%
\begin{equation}
\omega_{X_{1}\cdots X_{K}C}\equiv\sum_{x_{1},\ldots,x_{K}}p_{X_{1}}%
(x_{1})\cdots p_{X_{K}}(x_{K})|x_{1}\rangle\langle x_{1}|_{X_{1}}\otimes
\cdots\otimes|x_{K}\rangle\langle x_{K}|_{X_{K}}\otimes\rho_{x_{1}%
,\ldots,x_{K}}.
\end{equation}
Observe that $H_{2}(A|Y)_{\sigma}=-\log\sum_{y}p(y)\operatorname{Tr}%
\{\sigma_{y}^{2}\}$ for a state $\sigma_{YA}=\sum_{y}p(y)|y\rangle\langle
y|_{Y}\otimes\sigma_{y}$.
\end{theorem}

The rate region given in Theorem~\ref{thm:mult-send-gen-alt} is arguably an
improvement over that from \cite[Theorem~6]{fawzi2012classical}. By using the
coding technique from \cite[Theorem~6]{fawzi2012classical}, for an $m$-sender
multiple access channel, only $m$ of the inequalities feature the conditional
von Neumann entropy as the first term in the entropy differences, whereas the
other $2^{m}-m-1$ inequalities feature the conditional min-entropy. On the
other hand, all $2^{m}-1$ inequalities in Theorem~\ref{thm:mult-send-gen-alt}
feature the collision conditional entropy (conditional R\'{e}nyi entropy of
order two), which is never smaller than the conditional min-entropy. Thus, as
the number $m$\ of senders grows larger, the volume of the achievable rate
region from Theorem~\ref{thm:mult-send-gen-alt} is generally significantly
larger than the volume of the achievable rate region from \cite[Theorem~6]%
{fawzi2012classical}.

We can also compare the rate regions from Theorems~\ref{thm:cc-mac-renyi-gen}
and \ref{thm:mult-send-gen-alt}. By using the identity%
\begin{align}
&  H_{2}(X(\mathcal{J}^{c})C)_{\omega}-H(X(\mathcal{J}^{c})C|X(\mathcal{J}%
))_{\omega}\nonumber\\
&  =H_{2}(X(\mathcal{J}^{c})C)_{\omega}-H(C|X_{1}\cdots X_{K})_{\omega
}-H(X(\mathcal{J}^{c})|X(\mathcal{J}))_{\omega}\\
&  =H_{2}(X(\mathcal{J}^{c})C)_{\omega}-H(C|X_{1}\cdots X_{K})_{\omega
}-H(X(\mathcal{J}^{c}))_{\omega},
\end{align}
the difference between the information-theoretic terms in the inequalities in
\eqref{eq:info-bound-2-mac}\ and \eqref{eq:info-bound-1-mac} is given by%
\begin{equation}
H_{2}(C|X(\mathcal{J}^{c}))_{\omega}-H_{2}(X(\mathcal{J}^{c})C)_{\omega
}+H(X(\mathcal{J}^{c}))_{\omega}.
\end{equation}
The above difference can sometimes be negative or positive, as one can find by
some simple numerical tests with qubit states. Thus, the two rate regions from
Theorems~\ref{thm:cc-mac-renyi-gen} and \ref{thm:mult-send-gen-alt} are
generally incomparable. However, we can easily make use of both results:\ a
standard time-sharing argument establishes that the convex hull of the two
rates regions is achievable.

\begin{remark}
The quantum simultaneous decoding conjecture from
\cite{fawzi2012classical,Wil11} is the statement that the R\'{e}nyi entropies
of order two in Theorems~\ref{thm:cc-mac-renyi-gen} and
\ref{thm:mult-send-gen-alt} can be replaced by quantum entropies (i.e., von
Neumann entropies), while still employing a quantum simultaneous decoder. An
explicit statement of the quantum simultaneous decoding conjecture is
available in \cite[Conjecture~4]{fawzi2012classical}. The conjecture has been
solved in the case of two senders for unassisted classical communication in
\cite{fawzi2012classical}\ and for entanglement-assisted classical
communication in \cite{Xu2013}, but not for three or more senders in either case.
\end{remark}

\section{Quantum simultaneous decoding for multiple-access channels and
multiple quantum hypothesis testing}

\label{sec:disucssion}

In this section, we establish explicit links between the quantum simultaneous
decoding conjecture from \cite{fawzi2012classical,Wil11} and open questions
from \cite{AM14,BHOS14}\ in multiple quantum hypothesis testing. Recall that
the general goal of quantum hypothesis testing is to minimize the error
probability in identifying quantum states. In binary quantum hypothesis
testing, one considers two hypotheses: the null hypothesis is that a quantum
system is prepared in the state $\rho$, and the alternative hypothesis is that
the quantum system is prepared in the state $\sigma$, where $\rho,\sigma
\in\mathcal{D}(\mathcal{H})$. Operationally, the discriminator receives the
state $\rho$ with probability $p\in(0,1)$ and the state $\sigma$ with
probability $1-p$, and the task is to determine which state was prepared, by
means of some quantum measurement $\{T,I-T\}$, where the test operator $T$
satisfies $0\leq T\leq I$. There are two kinds of errors: a Type~I error
occurs when the state is identified as $\sigma$ when in fact $\rho$ was
prepared and a Type~II error is the opposite kind of error. The error
probabilities corresponding to the two types of errors are as follows:%
\begin{align}
\alpha(T,\rho)  &  \equiv\operatorname{Tr}\{(I-T)\rho\},\\
\beta(T,\sigma)  &  \equiv\operatorname{Tr}\{T\sigma\}~.
\end{align}

As in information theory, quantum hypothesis testing has been studied in the
asymptotic i.i.d.~setting. In the setting of \textit{symmetric} hypothesis
testing, we are interested in minimizing the overall error probability
\begin{align}
P_{e}^{\ast}(p\rho,(1-p)\sigma)  &  \equiv\inf_{T\,:\,0\leq T\leq
I}\ p\,\alpha(T,\rho)+(1-p)\beta(T,\sigma)\\
&  =\frac{1}{2}\left(  \operatorname{Tr}\{p\rho+(1-p)\sigma\}-\Vert
p\rho-(1-p)\sigma\Vert_{1}\right)  .
\end{align}
In the i.i.d.~setting, $n$ quantum systems are prepared as either
$\rho^{\otimes n}$ or $\sigma^{\otimes n}$, and the goal is to determine the
optimal exponent of the error probability, defined as
\begin{equation}
\lim_{n\rightarrow\infty}-\frac{1}{n}\log_{2}P_{e}^{\ast}(p\rho^{\otimes
n},(1-p)\sigma^{\otimes n}).
\end{equation}
One of the landmark results in quantum hypothesis testing is that the optimal
exponent is equal to the quantum Chernoff distance
\cite{ACMBMAV07,nussbaum2009chernoff}:
\begin{equation}
\lim_{n\rightarrow\infty}-\frac{1}{n}\log_{2}P_{e}^{\ast}(p\rho^{\otimes
n},(1-p)\sigma^{\otimes n})=C(\rho,\sigma)\equiv\sup_{s\in\lbrack0,1]}%
-\log_{2}\operatorname{Tr}\{\rho^{s}\sigma^{1-s}\}.
\end{equation}
This development can be generalized to the setting in which $\rho$, $\sigma$,
$p$, and $1-p$ can be replaced by positive semi-definite operators $\neq0$ and
positive constants. Indeed, for positive semi-definite $A$ and $B$, we can
define%
\begin{align}
P_{e}^{\ast}(A,B)  &  \equiv\inf_{T\,:\,0\leq T\leq I}\ \operatorname{Tr}%
\{(I-T)A\}+\operatorname{Tr}\{TB\}\label{eq:err-prob-positive-ops-1}\\
&  =\frac{1}{2}\left(  \operatorname{Tr}\{A+B\}-\Vert A-B\Vert_{1}\right)  .
\label{eq:err-prob-positive-ops}%
\end{align}
Then for positive constants $K_{0},K_{1}>0$, we find that%
\begin{equation}
\lim_{n\rightarrow\infty}-\frac{1}{n}\log_{2}P_{e}^{\ast}(K_{0}A^{\otimes
n},K_{1}B^{\otimes n})=C(A,B)\equiv\sup_{s\in\lbrack0,1]}-\log_{2}%
\operatorname{Tr}\{A^{s}B^{1-s}\}. \label{eq:q-chernoff-thm}%
\end{equation}

In the setting of \textit{asymmetric} hypothesis testing, we are interested in
the optimal exponent of the Type~II error $\beta(T,\sigma)$, under a
constraint on the Type~I error, i.e., $\alpha(T,\sigma)\leq\varepsilon$, with
$\varepsilon\in\lbrack0,1]$. That is, we are interested in the following
quantity, now known as hypothesis testing relative entropy:
\begin{equation}
D_{H}^{\varepsilon}(\rho\Vert\sigma)\equiv-\log_{2}\inf_{T}\{\beta(T,\sigma):
0\leq T\leq I \wedge\alpha(T,\sigma)\leq\varepsilon\}.
\end{equation}
The optimal exponential decay rate in the asymmetric setting is given by the
quantum Stein's lemma \cite{hiai1991proper,ogawa2000strong}, which establishes
the following for all $\varepsilon\in(0,1)$:%
\begin{equation}
\lim_{n\rightarrow\infty}\frac{1}{n}D_{H}^{\varepsilon}(\rho^{\otimes n}%
\Vert\sigma^{\otimes n})=D(\rho\Vert\sigma),
\end{equation}
giving the quantum relative entropy its fundamental operational interpretation.

As we can see from \cite{anshu2017one}\ and our developments in
Section~\ref{sec:EA-point-to-point}, position-based coding forges a direct
connection between single-sender single-receiver communication and binary
quantum hypothesis testing. Specifically, the Chernoff distance from symmetric
hypothesis testing gives a lower bound on the entanglement-assisted error
exponent; while the application of the results from asymmetric hypothesis
testing leads to a lower bound on the one-shot entanglement-assisted capacity
and in turn on the second-order coding rate for entanglement-assisted communication.

In what follows, we discuss the generalization of both asymmetric and
symmetric hypothesis testing to multiple quantum states and their connections
to quantum simultaneous decoding.

\subsection{Symmetric multiple quantum hypothesis testing and quantum
simultaneous decoding}

We now tie one version of the quantum simultaneous decoding conjecture to
\cite[Conjecture~4.2]{AM14}, which has to do with distinguishing one state
from a set of other possible states. To recall the setting of
\cite[Conjecture~4.2]{AM14}, suppose that a state $\rho$ is prepared with
probability $p\in(0,1)$ and with probability $1-p$ one state $\sigma_{i}$ of
$r$ states is prepared with probability $q_{i}$, where $i\in\{1,\ldots,r\}$.
The goal is to determine whether $\rho$ was prepared or whether one of the
other states was prepared, and the error probability in doing so is equal to%
\begin{equation}
P_{e}^{\ast}\!\left(  p\rho,(1-p)\sum_{i=1}^{r}q_{i}\sigma_{i}\right)  .
\end{equation}
The measurement operator that achieves the minimum error probability is equal
to
\begin{equation}
\left\{  p\rho-(1-p)\sum_{i=1}^{r}q_{i}\sigma_{i}\geq0\right\}  .
\end{equation}
As usual, we are interested in the i.i.d.~case, in which $\rho$ and
$\sigma_{i}$ are replaced by $\rho^{\otimes n}$ and $\sigma_{i}^{\otimes n}$
for large $n$, and \cite[Conjecture~4.2]{AM14} states that%
\begin{equation}
\lim_{n\rightarrow\infty}-\frac{1}{n}\log_{2}P_{e}^{\ast}\!\left(
p\rho^{\otimes n},(1-p)\sum_{i=1}^{r}q_{i}\sigma_{i}^{\otimes n}\right)
=\min_{i}C(\rho,\sigma_{i}).
\end{equation}
We now propose a slight generalization of \cite[Conjecture~4.2]{AM14} and
\eqref{eq:q-chernoff-thm}, in which $\rho$ and $\sigma_{i}$ are replaced by
positive semi-definite operators and $p$ and $1-p$ are replaced by positive constants:

\begin{conjecture}
\label{conj:simul-decode}Let $\{A,B_{1},\ldots,B_{r}\}$ be a set of positive
semi-definite operators with trace strictly greater than zero, and let $K_{0}%
$, \ldots, $K_{r}$ be strictly positive constants. Then the following equality
holds%
\begin{equation}
\lim_{n\rightarrow\infty}-\frac{1}{n}\log_{2}P_{e}^{\ast}\!\left(
K_{0}A^{{\otimes n}},\sum_{i=1}^{r}K_{i}B_{i}^{{\otimes n}}\right)  =\min
_{i}C(A,B_{i})~,
\end{equation}
where $P_{e}^{\ast}$ is defined in
\eqref{eq:err-prob-positive-ops-1}--\eqref{eq:err-prob-positive-ops} and
$C(A,B_{i})$ in \eqref{eq:q-chernoff-thm}.
\end{conjecture}

To see how Conjecture~\ref{conj:simul-decode}\ is connected to quantum
simultaneous decoding, recall from \eqref{eq:EA-MAC-one-shot}\ our bound on
the error probability when simultaneously decoding the message pair $(l,m)$:%
\begin{multline}
p_{e}(l,m)\leq4\operatorname{Tr}\{(I-T_{RSC})\mathcal{N}_{AB\rightarrow
C}(\theta_{RA}\otimes\gamma_{SB})\}+4\big[L\operatorname{Tr}\{T_{RSC}%
\mathcal{N}_{AB\rightarrow C}(\theta_{R}\otimes\theta_{A}\otimes\gamma
_{SB})\}\\
+M\operatorname{Tr}\{T_{RSC}\mathcal{N}_{AB\rightarrow C}(\theta_{RA}%
\otimes\gamma_{S}\otimes\gamma_{B})\}+LM\operatorname{Tr}\{T_{RSC}%
\mathcal{N}_{AB\rightarrow C}(\theta_{R}\otimes\theta_{A}\otimes\gamma
_{S}\otimes\gamma_{B})\}\big]~,
\end{multline}
where we have set $c=1$ and used that $L-1<L$ and $M-1<M$. Now applying this
bound to the i.i.d.~case and setting $L=2^{nR_{1}}$ and $M=2^{nR_{2}}$, we
find that the upper bound becomes%
\begin{equation}
p_{e}(l,m)\leq4\left[  \operatorname{Tr}\{(I-T)\rho^{\otimes n}%
\}+\operatorname{Tr}\{T\left(  B_{1}^{\otimes n}+B_{2}^{\otimes n}%
+B_{3}^{\otimes n}\right)  \}\right]  ,
\end{equation}
where%
\begin{align}
\rho &  \equiv\mathcal{N}_{AB\rightarrow C}(\theta_{RA}\otimes\gamma_{SB}),\\
B_{1}  &  \equiv2^{R_{1}}\mathcal{N}_{AB\rightarrow C}(\theta_{R}\otimes
\theta_{A}\otimes\gamma_{SB}),\\
B_{2}  &  \equiv2^{R_{2}}\mathcal{N}_{AB\rightarrow C}(\theta_{RA}%
\otimes\gamma_{S}\otimes\gamma_{B}),\\
B_{3}  &  \equiv2^{R_{1}+R_{2}}\mathcal{N}_{AB\rightarrow C}(\theta_{R}%
\otimes\theta_{A}\otimes\gamma_{S}\otimes\gamma_{B}).
\end{align}
To minimize the upper bound on the error probability, we should pick the test
operator $T$ as follows:%
\begin{equation}
T\equiv\left\{  \rho^{\otimes n}-\left(  B_{1}^{\otimes n}+B_{2}^{\otimes
n}+B_{3}^{\otimes n}\right)  \geq0\right\}  ,
\label{eq:q-simul-decoder-err-exp}%
\end{equation}
and then the upper bound becomes%
\begin{align}
p_{e}(l,m)  &  \leq2\left(  \operatorname{Tr}\{\rho^{\otimes n}+B_{1}^{\otimes
n}+B_{2}^{\otimes n}+B_{3}^{\otimes n}\}-\left\Vert \rho^{\otimes n}-\left(
B_{1}^{\otimes n}+B_{2}^{\otimes n}+B_{3}^{\otimes n}\right)  \right\Vert
_{1}\right) \\
&  =4\ P_{e}^{\ast}\!\left(  \rho^{{\otimes n}},B_{1}^{{\otimes n}}%
+B_{2}^{{\otimes n}}+B_{3}^{{\otimes n}}\right)  .
\end{align}
The test operator $T$ given in \eqref{eq:q-simul-decoder-err-exp} was
previously realized in \cite{Wil10note} and \cite{hayashi2017exponent} to be
relevant as a quantum simultaneous decoder in the context of unassisted
classical communication over a quantum multiple access channel.

Now applying Conjecture~\ref{conj:simul-decode} (provided it is true), we
would find that the error probability $p_{e}(l,m)$\ is bounded from above as
$p_{e}(l,m)\lesssim e^{-nE(R_{1},R_{2})}$, with the error exponent
$E(R_{1},R_{2})$ equal to%
\begin{align}
E(R_{1},R_{2})  &  =\min\{C(\rho,B_{1}),C(\rho,B_{2}),C(\rho,B_{3})\}\\
&  =\min\Big\lbrace\sup_{s\in\lbrack0,1]}-\log_{2}\operatorname{Tr}\{\rho
^{s}B_{1}^{1-s}\},\sup_{s\in\lbrack0,1]}-\log_{2}\operatorname{Tr}\{\rho
^{s}B_{2}^{1-s}\},\nonumber\\
&  \qquad\qquad\qquad\sup_{s\in\lbrack0,1]}-\log_{2}\operatorname{Tr}%
\{\rho^{s}B_{3}^{1-s}\}\Big\rbrace\\
&  =\min\Big\lbrace\sup_{s\in\lbrack0,1]}(1-s)\Big[I_{s}(R;CS)_{\omega}%
-R_{1}\Big],\sup_{s\in\lbrack0,1]}(1-s)\Big[I_{s}(S;CR)_{\omega}%
-R_{2}\Big],\nonumber\\
&  \qquad\qquad\qquad\sup_{s\in\lbrack0,1]}(1-s)\Big[I_{s}(RS;C)_{\omega
}-(R_{1}+R_{2})\Big]~. \label{eq:multiple-err-exp}%
\end{align}
Thus the rate region $(R_{1},R_{2})$ would be achievable as long as
$E(R_{1},R_{2})>0$. Now using the fact that, for a bipartite state $\rho
_{AB}$
\begin{equation}
\lim_{s\rightarrow1}I_{s}(A;B)_{\rho}=I(A;B)_{\rho}~,
\end{equation}
we would then find that the following rate region is achievable:%
\begin{align}
R_{1}  &  \leq I(R;CS)_{\omega}~,\label{eq:ach-simul-1}\\
R_{2}  &  \leq I(S;CR)_{\omega}~,\\
R_{1}+R_{2}  &  \leq I(RS;C)_{\omega}~, \label{eq:ach-simul-3}%
\end{align}
where $\omega_{RSC}=\mathcal{N}_{AB\rightarrow C}(\theta_{RA}\otimes
\gamma_{SB})$ and the above approach would solve the quantum simultaneous
decoding conjecture. The method clearly extends to more than two senders. We
remark that aspects of the above approach were discussed in the recent work
\cite{hayashi2017exponent} for the case of unassisted classical communication
over a quantum multiple access channel, but there the connection to
\cite[Conjecture~4.2]{AM14} or Conjecture~\ref{conj:simul-decode} was not
realized, nor was the entanglement-assisted case considered.

We end this section by noting that \cite[Theorem~4.3]{AM14} offers several
suboptimal upper bounds on the error probability $P_{e}^{\ast}\!\left(
\rho^{{\otimes n}},B_{1}^{{\otimes n}}+B_{2}^{{\otimes n}}+B_{3}^{{\otimes n}%
}\right)  $, which in turn could be used to establish suboptimal achievable
rate regions for the quantum multiple-access channel. However, here we refrain
from the details of what these regions would be, except to say that they would
be in terms of the negative logarithm of the fidelity, replacing $I_{s}$ in \eqref{eq:multiple-err-exp}.

\subsection{Asymmetric hypothesis testing with composite alternative
hypothesis}

\label{subsec:asymmetic}

We now tie the quantum simultaneous decoding problem to a different open
question in asymmetric hypothesis testing. Recall our upper bound from
\eqref{eq:EA-MAC-one-shot} on the error probability for classical
communication over a quantum multiple-access channel, as applied for the
i.i.d.~case:%
\begin{equation}
p_{e}(l,m)\leq c_{\operatorname{I}}\operatorname{Tr}\{(I-T)\rho^{{\otimes n}%
}\}+c_{\operatorname{II}}\Big[\operatorname{Tr}\{T\left[  B_{1}^{{\otimes n}%
}+B_{2}^{{\otimes n}}+B_{3}^{{\otimes n}}\right]  \}\Big]~,
\label{ieq:upper-bound-hypothesis}%
\end{equation}
where $L=2^{nR_{1}}$, $M=2^{nR_{2}}$, and the state $\rho$ and the positive
semi-definite operators $B_{1}$, $B_{2}$, and $B_{3}$ are given by%
\begin{align}
\rho &  =\mathcal{N}_{AB\rightarrow C}(\theta_{RA}\otimes\gamma_{SB}%
)~,\label{eq:rho-ashypo}\\
B_{1}  &  =2^{R_{1}}\mathcal{N}_{AB\rightarrow C}(\theta_{R}\otimes\theta
_{A}\otimes\gamma_{SB}),\\
B_{2}  &  =2^{R_{2}}\mathcal{N}_{AB\rightarrow C}(\theta_{RA}\otimes\gamma
_{S}\otimes\gamma_{B}),\\
B_{3}  &  =2^{R_{1}+R_{2}}\mathcal{N}_{AB\rightarrow C}(\theta_{R}%
\otimes\theta_{A}\otimes\gamma_{S}\otimes\gamma_{B}). \label{eq:sig-3-ashypo}%
\end{align}
Rather than try to minimize the overall error probability as we did in the
previous section, we could try to minimize all of the other error
probabilities subject to a constraint on the error probability
$\operatorname{Tr}\{(I-T)\rho^{{\otimes n}}\}$. Thus, we seek a test operator
$T$ which is capable of discriminating the state $\rho^{\otimes n}$ from the
operator $B_{1}^{{\otimes n}}+B_{2}^{{\otimes n}}+B_{3}^{{\otimes n}}$. This
kind of task is formally called asymmetric hypothesis testing with composite
alternative hypothesis. The problem of a composite null hypothesis is that
considered in the context of the quantum Sanov theorem, which was solved in
\cite{bjelakovic2005quantum} (see also \cite{H02}), and finds application in
communication over compound channels \cite{BB09,mosonyi2015coding} (see also
\cite{DD07,Hayashi2009}).

The following open question, strongly related to a question from
\cite{BHOS14}, is relevant for asymmetric hypothesis testing with composite
alternative hypothesis:

\begin{question}
\label{conj:asymmetric} Consider a quantum state $\rho\in\mathcal{D}%
(\mathcal{H})$, a positive integer $r>1$, and a finite set of positive
semi-definite operators $\mathcal{B}\equiv\{B_{i}:1\leq i\leq r\}$, for which
$\operatorname{supp}(\rho)\subseteq\operatorname{supp}(B_{i})$ $\forall
B_{i}\in\mathcal{B}$ and%
\begin{equation}
\min_{i}D(\rho\Vert B_{i}) > 0.
\end{equation}
What is the most general form that $\rho$ and $\mathcal{B}$ can take such that
the following statement is true? For all $\varepsilon\in(0,1)$, $\delta>0$,
and sufficiently large $n$, there exists a binary test $\{T,I-T\}$ such that
the Type~I error is bounded from above by $\varepsilon$:%
\begin{equation}
\operatorname{Tr}\{(I-T)\rho^{{\otimes n}}\}\leq\varepsilon~,
\end{equation}
and for all $B_{i}\in\mathcal{B}$, the exponential decay rate of the Type~II
error is bounded from below as follows:
\begin{equation}
-\frac{1}{n}\log_{2}\operatorname{Tr}\{TB_{i}^{{\otimes n}}\}\geq\left[
\min_{i}D(\rho\Vert B_{i})\right]  -\delta~.
\end{equation}

\end{question}

Below we prove the following special case:

\begin{theorem}
The statement at the end of Question~\ref{conj:asymmetric} is true when the
set $\mathcal{B}$ forms a commuting set of operators (each of which need not
commute with $\rho$).
\end{theorem}

\begin{proof}
To this end, we employ the notion of a relative typical projector, which was
used in \cite{BS12} to establish an alternate proof of the quantum Stein lemma
(see also \cite{BLW15} for a different use of relative typical projectors).
Let $B_{i}=\sum_{y}f_{Y}^{i}(y)|\phi_{y}^{i}\rangle\langle\phi_{y}^{i}|$
denote a spectral decomposition of $B_{i}$. For a state $\rho$ and positive
semi-definite operator $B_{i}$, define the relative typical subspace
$T_{\rho|B_{i}}^{\delta,n}$ for $\delta>0$ and integer $n\geq1$ as%
\begin{equation}
T_{\rho|B_{i}}^{\delta,n}\equiv\mathrm{span}\left\{  |\phi_{y^{n}}^{i}%
\rangle:\left\vert -\frac{1}{n}\log_{2}(f_{Y^{n}}^{i}(y^{n}%
))+\operatorname{Tr}\{\rho\log_{2}B_{i}\}\right\vert \leq\delta\right\}  ,
\label{qtyp}%
\end{equation}
where%
\begin{align}
y^{n}  &  \equiv y_{1}\cdots y_{n},\\
f_{Y^{n}}^{i}\left(  y^{n}\right)   &  \equiv\prod\limits_{j=1}^{n}f_{Y}%
^{i}\left(  y_{j}\right)  ,\\
|\phi_{y^{n}}^{i}\rangle &  \equiv|\phi_{y_{1}}^{i}\rangle\otimes\cdots
\otimes|\phi_{y_{n}}^{i}\rangle.
\end{align}
Let $\Pi_{\rho|B_{i},\delta}^{n}$ denote the projection operator corresponding
to the relative typical subspace $T_{\rho|B_{i}}^{\delta,n}$. The critical
properties of the relative typical projector are as follows:%
\begin{align}
\operatorname{Tr}\{\Pi_{\rho|B_{i},\delta}^{n}\rho^{\otimes n}\}  &
\geq1-\varepsilon,\\
2^{-n\left[  -\operatorname{Tr}\{\rho\log_{2}B_{i}\}+\delta\right]  }\Pi
_{\rho|B_{i},\delta}^{n}  &  \leq\Pi_{\rho|B_{i},\delta}^{n}B_{i}^{\otimes
n}\Pi_{\rho|B_{i},\delta}^{n}\leq2^{-n\left[  -\operatorname{Tr}\{\rho\log
_{2}B_{i}\}-\delta\right]  }\Pi_{\rho|B_{i},\delta}^{n},
\end{align}
with the first inequality holding for all $\varepsilon\in(0,1)$, $\delta>0$,
and sufficiently large $n$.

The main idea for the proof under the stated assumptions is to take the test
operator $T$ as%
\begin{equation}
T=\Pi_{\rho|B_{r},\delta}^{n}\cdots\Pi_{\rho|B_{1},\delta}^{n}\Pi_{\rho
,\delta}^{n}\Pi_{\rho|B_{1},\delta}^{n}\cdots\Pi_{\rho|B_{r},\delta}^{n},
\end{equation}
where $\Pi_{\rho,\delta}^{n}$ is the typical projector for $\rho$. Then we
find that for all $\varepsilon\in(0,1)$, $\delta>0$, and sufficiently large
$n$%
\begin{align}
\operatorname{Tr}\{T\rho^{\otimes n}\}  &  \geq\operatorname{Tr}\{\Pi
_{\rho,\delta}^{n}\rho^{\otimes n}\}-\sum_{i=1}^{r}\left\Vert \Pi_{\rho
|B_{i},\delta}^{n}\rho^{\otimes n}\Pi_{\rho|B_{i},\delta}^{n}-\rho^{\otimes
n}\right\Vert _{1}\\
&  \geq1-\varepsilon-2r\sqrt{\varepsilon},
\end{align}
which follows by applying Lemmas~\ref{lemma:gentle} and \ref{lemma:close} and
properties of typicality and relative typicality.\ To handle the other kind of
error, consider that, from the assumption, all of the projectors $\Pi
_{\rho|B_{i},\delta}^{n}$ commute, so that%
\begin{align}
\operatorname{Tr}\{TB_{i}^{\otimes n}\}  &  =\operatorname{Tr}\{\Pi
_{\rho|B_{r},\delta}^{n}\cdots\Pi_{\rho|B_{1},\delta}^{n}\Pi_{\rho,\delta}%
^{n}\Pi_{\rho|B_{1},\delta}^{n}\cdots\Pi_{\rho|B_{r},\delta}^{n}B_{i}^{\otimes
n}\}\\
&  =\operatorname{Tr}\{\Pi_{\rho,\delta}^{n}\Pi_{\rho|B_{1},\delta}^{n}%
\cdots\Pi_{\rho|B_{r},\delta}^{n}B_{i}^{\otimes n}\Pi_{\rho|B_{r},\delta}%
^{n}\cdots\Pi_{\rho|B_{1},\delta}^{n}\}\\
&  =\operatorname{Tr}\{\Pi_{\rho,\delta}^{n}\Pi_{\rho|B_{1},\delta}^{n}%
\cdots\Pi_{\rho|B_{r},\delta}^{n}\Pi_{\rho|B_{i},\delta}^{n}B_{i}^{\otimes
n}\Pi_{\rho|B_{i},\delta}^{n}\Pi_{\rho|B_{r},\delta}^{n}\cdots\Pi_{\rho
|B_{1},\delta}^{n}\}\\
&  \leq2^{-n\left[  -\operatorname{Tr}\{\rho\log_{2}B_{i}\}-\delta\right]
}\operatorname{Tr}\{\Pi_{\rho,\delta}^{n}\Pi_{\rho|B_{1},\delta}^{n}\cdots
\Pi_{\rho|B_{r},\delta}^{n}\Pi_{\rho|B_{i},\delta}^{n}\Pi_{\rho|B_{r},\delta
}^{n}\cdots\Pi_{\rho|B_{1},\delta}^{n}\}\\
&  \leq2^{-n\left[  -\operatorname{Tr}\{\rho\log_{2}B_{i}\}-\delta\right]
}\operatorname{Tr}\{\Pi_{\rho,\delta}^{n}\}\\
&  \leq2^{-n\left[  -\operatorname{Tr}\{\rho\log_{2}B_{i}\}-\delta\right]
}2^{n\left[  H(\rho)+\delta\right]  }\\
&  =2^{-n\left[  D(\rho\Vert B_{i})-2\delta\right]  }.
\end{align}
The statement of the theorem follows by setting $\varepsilon^{\prime}%
\equiv\varepsilon+2r\sqrt{\varepsilon}$ and $\delta^{\prime}\equiv2\delta$,
considering that we have shown the existence of a test $T$ for which%
\begin{equation}
\operatorname{Tr}\{T\rho^{\otimes n}\}\geq1-\varepsilon^{\prime}%
,\qquad\operatorname{Tr}\{TB_{i}^{\otimes n}\}\leq2^{-n\left[  D(\rho\Vert
B_{i})-\delta^{\prime}\right]  },
\end{equation}
and it is possible to satisfy this for any choice of $\varepsilon^{\prime}%
\in(0,1)$ and $\delta^{\prime}>0$ by taking $n$ sufficiently large. (Note that
for the bound on the second kind of error probability to be decaying
exponentially, we require $\delta^{\prime}>0$ to be small enough so that
$\min_{i}D(\rho\Vert B_{i}) > \delta^{\prime}$.) We remark that this
conclusion is actually stronger than what is stated in
Question~\ref{conj:asymmetric} because here we conclude for all $B_{i}%
\in\mathcal{B}$, that%
\begin{equation}
-\frac{1}{n}\log_{2}\operatorname{Tr}\{TB_{i}^{{\otimes n}}\}\geq D(\rho\Vert
B_{i})-\delta^{\prime}\geq\left[  \min_{i}D(\rho\Vert B_{i})\right]
-\delta^{\prime}~.
\end{equation}
This concludes the proof.
\end{proof}

\bigskip

To see how Question~\ref{conj:asymmetric} is related to quantum simultaneous
decoding of the multiple-access channel, consider that for $\rho$, $B_{1}$,
\ldots, $B_{3}$ as defined in \eqref{eq:rho-ashypo}--\eqref{eq:sig-3-ashypo},
the inequality%
\begin{equation}
\min_{i}D(\rho\Vert B_{i})>0
\end{equation}
is equivalent to the following set of inequalities:%
\begin{align}
R_{1}  &  < I(R;CS)_{\omega}~,\\
R_{2}  &  < I(S;CR)_{\omega}~,\\
R_{1}+R_{2}  &  < I(RS;C)_{\omega}~,
\end{align}
where $\omega_{RSC}=\mathcal{N}_{AB\rightarrow C}(\theta_{RA}\otimes
\gamma_{SB})$. This equivalence holds because $\min_{i}D(\rho\Vert B_{i})>0$
is equivalent to the following three inequalities:%
\begin{equation}
D(\rho\Vert B_{1})>0,\qquad D(\rho\Vert B_{2})>0,\qquad D(\rho\Vert B_{3})>0,
\end{equation}
and%
\begin{align}
D(\rho\Vert B_{1})  &  =D(\mathcal{N}_{AB\rightarrow C}(\theta_{RA}%
\otimes\gamma_{SB})\Vert2^{R_{1}}\mathcal{N}_{AB\rightarrow C}(\theta
_{R}\otimes\theta_{A}\otimes\gamma_{SB}))\\
&  =D(\mathcal{N}_{AB\rightarrow C}(\theta_{RA}\otimes\gamma_{SB}%
)\Vert\mathcal{N}_{AB\rightarrow C}(\theta_{R}\otimes\theta_{A}\otimes
\gamma_{SB}))-R_{1}\\
&  =I(R;CS)_{\omega}-R_{1},\\
D(\rho\Vert B_{2})  &  =D(\mathcal{N}_{AB\rightarrow C}(\theta_{RA}%
\otimes\gamma_{SB})\Vert2^{R_{2}}\mathcal{N}_{AB\rightarrow C}(\theta
_{RA}\otimes\gamma_{S}\otimes\gamma_{B}))\\
&  =D(\mathcal{N}_{AB\rightarrow C}(\theta_{RA}\otimes\gamma_{SB}%
)\Vert\mathcal{N}_{AB\rightarrow C}(\theta_{RA}\otimes\gamma_{S}\otimes
\gamma_{B}))-R_{2},\\
&  =I(S;CR)_{\omega}-R_{2},\\
D(\rho\Vert B_{3})  &  =D(\mathcal{N}_{AB\rightarrow C}(\theta_{RA}%
\otimes\gamma_{SB})\Vert2^{R_{1}+R_{2}}\mathcal{N}_{AB\rightarrow C}%
(\theta_{R}\otimes\theta_{A}\otimes\gamma_{SB}))\\
&  =D(\mathcal{N}_{AB\rightarrow C}(\theta_{RA}\otimes\gamma_{SB}%
)\Vert\mathcal{N}_{AB\rightarrow C}(\theta_{R}\otimes\theta_{A}\otimes
\gamma_{SB}))-(R_{1}+R_{2})\\
&  =I(RS;C)_{\omega}-(R_{1}+R_{2}).
\end{align}
Thus, if such a sequence of test operators existed as stated in
Question~\ref{conj:asymmetric}, then the error probability $p_{e}(l,m)$ when
decoding a multiple-access channel could be bounded from above as%
\begin{align}
p_{e}(l,m)  &  \leq c_{\operatorname{I}}\operatorname{Tr}\{(I-T)\rho^{{\otimes
n}}\}+c_{\operatorname{II}}\Big[\operatorname{Tr}\{T\left[  B_{1}^{{\otimes
n}}+B_{2}^{{\otimes n}}+B_{3}^{{\otimes n}}\right]  \}\Big]\\
&  \leq c_{\operatorname{I}}\varepsilon+3c_{\operatorname{II}}%
\Big[2^{-n\left[  \min_{i}D(\rho\Vert B_{i})-\delta\right]  }\Big]~,
\end{align}
where $\rho$, $B_{1}$, \ldots, $B_{3}$ are defined in
\eqref{eq:rho-ashypo}--\eqref{eq:sig-3-ashypo}. Then by choosing the rates
$R_{1}$ and $R_{2}$ to satisfy%
\begin{align}
R_{1}+2\delta &  \leq I(R;CS)_{\omega}~,\\
R_{2}+2\delta &  \leq I(S;CR)_{\omega}~,\\
R_{1}+R_{2}+2\delta &  \leq I(RS;C)_{\omega}~,
\end{align}
which is equivalent to $\min_{i}D(\rho\Vert B_{i})\geq2\delta$, we would have%
\begin{equation}
p_{e}(l,m)\leq c_{\operatorname{I}}\varepsilon+3c_{\operatorname{II}%
}2^{-n\delta},
\end{equation}
and we could thus make the error probability as small as desired by taking $n$
sufficiently large. Since $\delta>0$ is arbitrary, we could then say that the
rate region in \eqref{eq:ach-simul-1}--\eqref{eq:ach-simul-3}\ would be
achievable. If the statement at the end of Question~\ref{conj:asymmetric}%
\ holds for the states given above, then the method would clearly lead to a
quantum simultaneous decoder for more than two senders, by a straightforward
generalization of the above approach.

The authors of \cite{brandao2010generalization} considered a similar problem
in asymmetric hypothesis testing for a specific family of states with certain
permutation symmetry. We should point out that in
\cite{brandao2010generalization}, the lower bound for the exponential rate of
the Type~II error is given by a regularized version of $\min_{i}D(\rho
\Vert\sigma_{i})$. If a similar result held, along the lines stated in
Question~\ref{conj:asymmetric} and related to the conjecture in \cite{BHOS14},
without the need for regularization and for the operators in
\eqref{eq:rho-ashypo}--\eqref{eq:sig-3-ashypo} (and more general ones relevant
for more senders), then the developments in the present paper would
immediately give bounds for the performance of quantum simultaneous decoding
for the quantum multiple-access channel.

We end this section by remarking that our quantum simultaneous decoder in
\eqref{eq:renyi-2-simul-decode}\ gives a method for distinguishing
$\rho^{\otimes n}$ from the operator $B_{1}^{{\otimes n}}+B_{2}^{{\otimes n}%
}+B_{3}^{{\otimes n}}$ with Type I\ error probability bounded, for
sufficiently large $n$, by an arbitrary $\varepsilon\in(0,1)$ and the Type II
error probability bounded, for a positive constant $c$, by%
\begin{equation}
\approx c\ 2^{-n\min\left\{  \tilde{I}(S;CR)_{\omega}-R_{1},\ \tilde
{I}(R;CS)_{\omega}-R_{2},\ \tilde{I}(RC;S)_{\omega}-\left(  R_{1}%
+R_{2}\right)  \right\}  }~,
\end{equation}
and%
\begin{equation}
\approx c\ 2^{-n\min\left\{  I^{\prime}(S;CR)_{\omega}-R_{1},\ I^{\prime
}(R;CS)_{\omega}-R_{2},\ I^{\prime}(RC;S)_{\omega}-\left(  R_{1}+R_{2}\right)
\right\}  },
\end{equation}
where $\omega_{RSC}=\mathcal{N}_{AB\rightarrow C}(\theta_{RA}\otimes
\gamma_{SB})$. To obtain the first statement, we apply the analysis in the
proof of Theorem~\ref{thm:achieve-simul-renyi-2}, and for the second, we apply
the analysis in the proof of Theorem~\ref{thm:alt-simul-decode}. The above
statements have clear generalizations to more systems by invoking
Theorems~\ref{thm:mult-send-gen}\ and~\ref{thm:alt-simul-decode}. Thus, our
previous analysis in the context of communication applies for this
interesting, special case of Question~\ref{conj:asymmetric}, albeit with
suboptimal rates.

\section{Conclusion}

\label{sec:conclusion}

In this paper, we apply position-based coding to establish bounds of various
quantities for classical communication. For entanglement-assisted classical
communication over point-to-point quantum channels, we establish lower bounds
on the one-shot error exponent, the one-shot capacity, and the second-order
coding rate. We also find an alternative proof for an upper bound on one-shot
entanglement-assisted classical capacity, which is arguably simpler than the
approach from \cite{matthews2014finite}. We give an achievable rate region for
entanglement-assisted classical communication over multiple-access quantum
channels. Furthermore, we explicitly show how to derandomize a
randomness-assisted protocol (for multiple-access channel) to one without
assistance from any extra resources. Our results indicate that position-based
coding can be a powerful tool in achievability proofs of various communication
protocols in quantum Shannon theory (for recent applications of position-based
coding to private classical communication, see \cite{wilde2017position}). We
finally tied some open questions in multiple quantum hypothesis testing to the
quantum simultaneous decoding conjecture. Thus, we have shown that open
problems in multiple quantum hypothesis testing are fundamental to the study
of classical information transmission over quantum multiple-access channels.

\bigskip

\textbf{Acknowledgements.} We are grateful to Anurag Anshu, Mario Berta,
Fernando Brandao, Masahito Hayashi, Rahul Jain, and Marco Tomamichel for
discussions related to the topic of this paper. HQ acknowledges support from
the Air Force Office of Scientific Research, the Army Research Office, the
National Science Foundation, and the Northrop Grumman Corporation. QLW is
supported by the NFSC (Grants No.~61272057, No.~61309029 and No.~61572081) and
funded by the China Scholarship Council (Grant No.~201506470043).
MMW\ acknowledges support from the Office of Naval Research and the National
Science Foundation.

\appendix

\section{Proof of Proposition~\ref{prop:ineq-hypo-renyi} and quantum Stein's
lemma}

Here we provide a proof of Proposition~\ref{prop:ineq-hypo-renyi}. After doing
that, we discuss briefly how Proposition~\ref{prop:ineq-hypo-renyi} and
\cite[Lemma~5]{CMW14}\ lead to a transparent proof of the quantum Stein's
lemma \cite{hiai1991proper,ogawa2000strong}.

\bigskip

\begin{proof}
[Proof of Proposition~\ref{prop:ineq-hypo-renyi}]The statement of
Proposition~\ref{prop:ineq-hypo-renyi} is trivially true if $\rho\sigma=0$,
since both $D_{H}^{\varepsilon}(\rho\Vert\sigma)=\infty$ and $D_{\alpha}%
(\rho\Vert\sigma)=\infty$ in this case. So we consider the non-trivial case
when this equality does not hold. We exploit Lemma~\ref{lemma:spectral-ineq}
to establish the above bound. The proof also bears some similarities to a
related proof in \cite{AJW17}. Recall from Lemma~\ref{lemma:spectral-ineq}
that the following inequality holds for positive semi-definite operators $A$
and $B$ and for $\alpha\in(0,1)$:%
\begin{align}
\inf_{T:0\leq T\leq I}\operatorname{Tr}\{(I-T)A\}+\operatorname{Tr}\{TB\}  &
=\frac{1}{2}\left(  \operatorname{Tr}\{A+B\}-\left\Vert A-B\right\Vert
_{1}\right) \\
&  \leq\operatorname{Tr}\{A^{\alpha}B^{1-\alpha}\}.
\end{align}
For $p\in(0,1)$, pick $A=p\rho$ and $B=\left(  1-p\right)  \sigma$. Plugging
in to the above inequality, we find that there exists a measurement operator
$T^{\ast}=T(p,\rho,\sigma)$ such that%
\begin{equation}
p\operatorname{Tr}\{(I-T^{\ast})\rho\}+(1-p)\operatorname{Tr}\{T^{\ast}%
\sigma\}\leq p^{\alpha}(1-p)^{1-\alpha}\operatorname{Tr}\{\rho^{\alpha}%
\sigma^{1-\alpha}\}.
\end{equation}
This implies that%
\begin{equation}
p\operatorname{Tr}\{(I-T^{\ast})\rho\}\leq p^{\alpha}(1-p)^{1-\alpha
}\operatorname{Tr}\{\rho^{\alpha}\sigma^{1-\alpha}\},
\end{equation}
and in turn that%
\begin{equation}
\operatorname{Tr}\{(I-T^{\ast})\rho\}\leq\left(  \frac{1-p}{p}\right)
^{1-\alpha}\operatorname{Tr}\{\rho^{\alpha}\sigma^{1-\alpha}\}.
\end{equation}
For a given $\varepsilon\in(0,1)$ and $\alpha\in(0,1)$, we pick $p\in(0,1)$
such that%
\begin{equation}
\left(  \frac{1-p}{p}\right)  ^{1-\alpha}\operatorname{Tr}\{\rho^{\alpha
}\sigma^{1-\alpha}\}=\varepsilon.
\end{equation}
This is possible because we can rewrite the above equation as%
\begin{align}
\varepsilon &  =\left(  \frac{1-p}{p}\right)  ^{1-\alpha}\operatorname{Tr}%
\{\rho^{\alpha}\sigma^{1-\alpha}\}=\left(  \frac{1}{p}-1\right)  ^{1-\alpha
}\operatorname{Tr}\{\rho^{\alpha}\sigma^{1-\alpha}\}\\
\Leftrightarrow\left(  \frac{1}{p}-1\right)  ^{1-\alpha}  &  =\frac
{\varepsilon}{\operatorname{Tr}\{\rho^{\alpha}\sigma^{1-\alpha}\}}\\
\Leftrightarrow\frac{1}{p}  &  =\left[  \frac{\varepsilon}{\operatorname{Tr}%
\{\rho^{\alpha}\sigma^{1-\alpha}\}}\right]  ^{1/\left(  1-\alpha\right)  }+1\\
\Leftrightarrow p  &  =\frac{1}{\left[  \varepsilon/\operatorname{Tr}%
\{\rho^{\alpha}\sigma^{1-\alpha}\}\right]  ^{1/\left(  1-\alpha\right)  }%
+1}\in\left(  0,1\right)  .
\end{align}
This means that $T(p,\rho,\sigma)$ with $p$ selected as above is a measurement
such that%
\begin{equation}
\operatorname{Tr}\{(I-T^{\ast})\rho\}\leq\varepsilon.
\end{equation}
Now using the fact that the measurement $T^{\ast\ast}$\ for the hypothesis
testing relative entropy achieves the smallest type II\ error probability (by
definition)\ and the fact that%
\begin{equation}
(1-p)\operatorname{Tr}\{T^{\ast}\sigma\}\leq p^{\alpha}(1-p)^{1-\alpha
}\operatorname{Tr}\{\rho^{\alpha}\sigma^{1-\alpha}\}
\end{equation}
implies%
\begin{equation}
\operatorname{Tr}\{T^{\ast}\sigma\}\leq\left(  \frac{p}{1-p}\right)  ^{\alpha
}\operatorname{Tr}\{\rho^{\alpha}\sigma^{1-\alpha}\},
\end{equation}
we find that%
\begin{equation}
\operatorname{Tr}\{T^{\ast\ast}\sigma\}\leq\left(  \frac{p}{1-p}\right)
^{\alpha}\operatorname{Tr}\{\rho^{\alpha}\sigma^{1-\alpha}\}.
\end{equation}
Considering that%
\begin{equation}
\varepsilon=\left(  \frac{1-p}{p}\right)  ^{1-\alpha}\operatorname{Tr}%
\{\rho^{\alpha}\sigma^{1-\alpha}\}=\left(  \frac{p}{1-p}\right)  ^{\alpha
-1}\operatorname{Tr}\{\rho^{\alpha}\sigma^{1-\alpha}\}
\end{equation}
implies that%
\begin{equation}
\left[  \frac{\varepsilon}{\operatorname{Tr}\{\rho^{\alpha}\sigma^{1-\alpha
}\}}\right]  ^{1/(\alpha-1)}=\frac{p}{1-p},
\end{equation}
we get that%
\begin{align}
\operatorname{Tr}\{T^{\ast\ast}\sigma\}  &  \leq\left(  \frac{p}{1-p}\right)
^{\alpha}\operatorname{Tr}\{\rho^{\alpha}\sigma^{1-\alpha}\}\\
&  =\left(  \left[  \frac{\varepsilon}{\operatorname{Tr}\{\rho^{\alpha}%
\sigma^{1-\alpha}\}}\right]  ^{1/(\alpha-1)}\right)  ^{\alpha}%
\operatorname{Tr}\{\rho^{\alpha}\sigma^{1-\alpha}\}\\
&  =\varepsilon^{\alpha/\left(  \alpha-1\right)  }\left[  \operatorname{Tr}%
\{\rho^{\alpha}\sigma^{1-\alpha}\}\right]  ^{\alpha/(1-\alpha)}%
\operatorname{Tr}\{\rho^{\alpha}\sigma^{1-\alpha}\}\\
&  =\varepsilon^{\alpha/\left(  \alpha-1\right)  }\left[  \operatorname{Tr}%
\{\rho^{\alpha}\sigma^{1-\alpha}\}\right]  ^{1/(1-\alpha)}.
\end{align}
Then, by taking a logarithm, we get that%
\begin{align}
-\log_{2}\operatorname{Tr}\{T^{\ast\ast}\sigma\}  &  \geq-\log_{2}\!\left(
\varepsilon^{\alpha/\left(  \alpha-1\right)  }\left[  \operatorname{Tr}%
\{\rho^{\alpha}\sigma^{1-\alpha}\}\right]  ^{1/(1-\alpha)}\right) \\
&  =-\frac{\alpha}{\alpha-1}\log_{2}(\varepsilon)+\frac{1}{\alpha-1}\log
_{2}\operatorname{Tr}\{\rho^{\alpha}\sigma^{1-\alpha}\}\\
&  =-\frac{\alpha}{\alpha-1}\log_{2}(\varepsilon)+D_{\alpha}(\rho\Vert\sigma).
\end{align}
Putting everything together, we conclude the statement of
Proposition~\ref{prop:ineq-hypo-renyi}.
\end{proof}

\bigskip

We now briefly discuss how Proposition~\ref{prop:ineq-hypo-renyi} and
\cite[Lemma~5]{CMW14}, once established,\ lead to a transparent proof of the
quantum Stein's lemma \cite{hiai1991proper,ogawa2000strong} (see also
\cite{hayashi2007error}\ in this context). Before doing so, let us recall that
the quantum Stein's lemma (with strong converse) can be summarized as the
following equality holding for all $\varepsilon\in(0,1)$, states $\rho$, and
positive semi-definite operators $\sigma$:%
\begin{equation}
\lim_{n\rightarrow\infty}\frac{1}{n}D_{H}^{\varepsilon}(\rho^{\otimes n}%
\Vert\sigma^{\otimes n})=D(\rho\Vert\sigma),
\end{equation}
thus giving the quantum relative entropy its most fundamental operational
meaning as the optimal Type\ II\ error exponent in asymmetric quantum
hypothesis testing. Before giving the transparent proof, let us recall the
sandwiched R\'{e}nyi relative entropy \cite{MDSFT13,WWY13}, defined for
$\alpha\in(1,\infty)$ as%
\begin{equation}
\widetilde{D}_{\alpha}(\rho\Vert\sigma)\equiv\frac{1}{\alpha-1}\log
_{2}\operatorname{Tr}\{(\sigma^{(1-\alpha)/2\alpha}\rho\sigma^{(1-\alpha
)/2\alpha})^{\alpha}\}.
\end{equation}
whenever $\operatorname{supp}(\rho)\subseteq\operatorname{supp}(\sigma)$ and
set to $+\infty$ otherwise. For $\alpha\in(0,1)$, it is defined as above. The
sandwiched R\'{e}nyi relative entropy obeys the following limit
\cite{MDSFT13,WWY13}: $\lim_{\alpha\rightarrow1}\widetilde{D}_{\alpha}%
(\rho\Vert\sigma)=D(\rho\Vert\sigma)$. \cite[Lemma~5]{CMW14} is the statement
that the following inequality holds for all $\alpha>1$ and $\varepsilon
\in(0,1)$:%
\begin{equation}
D_{H}^{\varepsilon}(\rho\Vert\sigma)\leq\widetilde{D}_{\alpha}(\rho\Vert
\sigma)+\frac{\alpha}{\alpha-1}\log_{2}\!\left(  \frac{1}{1-\varepsilon
}\right)  .
\end{equation}

Employing Proposition~\ref{prop:ineq-hypo-renyi} and \cite[Lemma~5]{CMW14}
leads to a direct proof of the quantum Stein's lemma
\cite{hiai1991proper,ogawa2000strong}. Applying
Proposition~\ref{prop:ineq-hypo-renyi}, we find that the following inequality
holds for all $\alpha\in(0,1)$ and $\varepsilon\in(0,1)$:%
\begin{align}
\frac{1}{n}D_{H}^{\varepsilon}(\rho^{\otimes n}\Vert\sigma^{\otimes n})  &
\geq\frac{\alpha}{n(\alpha-1)}\log_{2}\!\left(  \frac{1}{\varepsilon}\right)
+\frac{1}{n}D_{\alpha}(\rho^{\otimes n}\Vert\sigma^{\otimes n})\\
&  =\frac{\alpha}{n(\alpha-1)}\log_{2}\!\left(  \frac{1}{\varepsilon}\right)
+D_{\alpha}(\rho\Vert\sigma). \label{eq:stein-lower}%
\end{align}
Taking the limit as $n\rightarrow\infty$ gives the following inequality
holding for all $\alpha\in(0,1)$:%
\begin{equation}
\lim_{n\rightarrow\infty}\frac{1}{n}D_{H}^{\varepsilon}(\rho^{\otimes n}%
\Vert\sigma^{\otimes n})\geq D_{\alpha}(\rho\Vert\sigma).
\end{equation}
We can then take the limit as $\alpha\rightarrow1$ to get that%
\begin{equation}
\lim_{n\rightarrow\infty}\frac{1}{n}D_{H}^{\varepsilon}(\rho^{\otimes n}%
\Vert\sigma^{\otimes n})\geq D(\rho\Vert\sigma).
\end{equation}
Applying \cite[Lemma~5]{CMW14}, we find that the following holds for all
$\alpha>1$ and $\varepsilon\in(0,1)$:%
\begin{align}
\frac{1}{n}D_{H}^{\varepsilon}(\rho^{\otimes n}\Vert\sigma^{\otimes n})  &
\leq\frac{\alpha}{n(\alpha-1)}\log_{2}\!\left(  \frac{1}{1-\varepsilon
}\right)  +\frac{1}{n}\widetilde{D}_{\alpha}(\rho^{\otimes n}\Vert
\sigma^{\otimes n})\\
&  =\frac{\alpha}{n(\alpha-1)}\log_{2}\!\left(  \frac{1}{1-\varepsilon
}\right)  +\widetilde{D}_{\alpha}(\rho\Vert\sigma). \label{eq:stein-upper}%
\end{align}
Taking the limit $n\rightarrow\infty$, we find that the following holds for
all $\alpha>1$%
\begin{equation}
\lim_{n\rightarrow\infty}\frac{1}{n}D_{H}^{\varepsilon}(\rho^{\otimes n}%
\Vert\sigma^{\otimes n})\leq\widetilde{D}_{\alpha}(\rho\Vert\sigma).
\end{equation}
Then taking the limit as $\alpha\rightarrow1$, we get that%
\begin{equation}
\lim_{n\rightarrow\infty}\frac{1}{n}D_{H}^{\varepsilon}(\rho^{\otimes n}%
\Vert\sigma^{\otimes n})\leq D(\rho\Vert\sigma).
\end{equation}

We note here that a slightly different approach would be to set $\alpha
=1+1/\sqrt{n}$ in both \eqref{eq:stein-lower} and \eqref{eq:stein-upper} and
then take the limit $n\rightarrow\infty$.

\bibliographystyle{alpha}
\bibliography{PBcode}

\end{document}